\documentclass[11pt]{article}
\usepackage{color,amsmath, amsthm, amsfonts, amssymb, dsfont, graphicx, listings, graphics, epsfig, enumerate, bbm,dsfont,mathrsfs}

\usepackage[margin=1in]{geometry}
\usepackage[fleqn,tbtags]{mathtools}
\usepackage[english]{babel}

\usepackage[round]{natbib} 

\title{Portfolio optimization with two coherent risk measures}
\author{Tahsin Deniz Akt\"{u}rk\thanks{The University of Chicago, Booth School of Business, Chicago, IL, USA, akturk@chicagobooth.edu}\and\c{C}a\u{g}{\i}n Ararat\thanks{Bilkent University, Department of Industrial Engineering, Ankara, Turkey, cararat@bilkent.edu.tr.}
}
\date{July 11, 2020}

\addto\captionsenglish {}
\makeatletter \renewenvironment{proof}[1][\proofname] {\par\pushQED{\qed}\normalfont\topsep6\p@\@plus6\p@\relax\trivlist\item[\hskip\labelsep\bfseries#1\@addpunct{.}]\ignorespaces}{\popQED\endtrivlist\@endpefalse} \makeatother
\newtheorem{theorem}{Theorem}[section]

\newtheorem{lemma}[theorem]{Lemma}
\newtheorem{proposition}[theorem]{Proposition}
\newtheorem{assumption}[theorem]{Assumption}

\theoremstyle{definition}
\newtheorem{example}[theorem]{Example}
\newtheorem{remark}[theorem]{Remark}

\numberwithin{equation}{section}
\newcommand{\R}{\mathbb{R}}
\newcommand{\sm}{\!\setminus\!}

\DeclareMathOperator{\qri}{qri}

\DeclareMathOperator{\co}{co}
\DeclareMathOperator{\cone}{cone}
\DeclareMathOperator{\interior}{int}

\DeclareMathOperator*{\argmax}{argmax}

\let\abs=\envert

\newcommand{\W}{\mathcal{W}}
\renewcommand{\O}{\Omega}
\renewcommand{\o}{\omega}

\newcommand{\F}{\mathcal{F}}
\newcommand{\D}{\mathscr{D}}

\newcommand{\X}{\mathcal{X}}
\newcommand{\Y}{\mathcal{Y}}

\newcommand{\A}{\mathscr{A}}
\renewcommand{\L}{\mathcal{L}}
\renewcommand{\P}{\mathscr{P}}
\newcommand{\E}{\mathbb{E}}
\newcommand{\V}{\text{Var}}
\newcommand{\M}{\mathscr{M}}
\newcommand{\N}{\mathcal{N}}
\renewcommand{\Pr}{\mathbb{P}}
\newcommand{\Q}{\mathbb{Q}}
\renewcommand{\a}{\alpha}
\renewcommand{\b}{\beta}
\newcommand{\g}{\gamma}
\renewcommand{\d}{\delta}

\renewcommand{\H}{\mathcal{H}}

\newcommand{\1}{\mathbf{1}}

\newcommand{\of}[1]{\ensuremath{\left( #1 \right)}}
\newcommand{\cb}[1]{\ensuremath{ \left\{ #1 \right\} }}
\newcommand{\sqb}[1]{\ensuremath{ \left[ #1 \right] }}

\newcommand{\ip}[1]{\ensuremath{ \left\langle #1 \right\rangle }}

\def\prehp(#1,#2){\ensuremath{  #1 \cdot #2 }}

\begin{document}
\maketitle
\thispagestyle{empty}

\begin{abstract}
We provide analytical results for a static portfolio optimization problem with two coherent risk measures. The use of two risk measures is motivated by joint decision-making for portfolio selection where the risk perception of the portfolio manager is of primary concern, hence, it appears in the objective function, and the risk perception of an external authority needs to be taken into account as well, which appears in the form of a risk constraint. The problem covers the risk minimization problem with an expected return constraint and the expected return maximization problem with a risk constraint, as special cases. For the general case of an arbitrary joint distribution for the asset returns, under certain conditions, we characterize the optimal portfolio as the optimal Lagrange multiplier associated to an equality-constrained dual problem. Then, we consider the special case of Gaussian returns for which it is possible to identify all cases where an optimal solution exists and to give an explicit formula for the optimal portfolio whenever it exists.\\
\\[-5pt]
\\[-5pt]
\end{abstract}

\section{Introduction}\label{intro}

The mean-variance portfolio selection problem introduced in the seminal work \cite{markowitz} is one of the most well-studied optimization problems. In the basic static version of the problem, one considers multiple correlated assets with known expected returns and covariances, and looks for an allocation of these assets. Considering the trade-off between the linear expected return and the quadratic variance, the problem can be formulated as a biobjective optimization problem whose efficient solutions form the so-called \emph{efficient frontier} on the mean-variance (or mean-standard deviation) plot of all portolios. \cite{merton} provides an analytical derivation of the efficient frontier for the general case of $n\geq 2$ assets.

The biobjective mean-variance problem can also be studied in terms of a parametric family of scalar (single-objective) problems. Among the popular \emph{scalarizations} are the ones where one minimizes variance over the set of all portfolios at a given expected return level, which is used as the parameter of the scalar problem. Analogously, one can impose a constraint on the variance using an upper bound parameter and maximizes expected return. Quite naturally, both approaches can be used to verify the analytical results in \cite{merton}.

Started with \cite{artzner}, the theory of \emph{coherent risk measures} provides an axiomatic approach to come up with functionals possessing desirable properties for risk measurement purposes. Such properties include monotonicity and translativity (see Section~\ref{rm} for precise definitions), which are not satisfied by variance or standard deviation. Canonical examples of coherent risk measures include (negative) expected value and average value-at-risk \citep{rockafellar2}. In addition, value-at-risk is also known to be a coherent risk measure when considered on a space of Gaussian random variables. Each of these three risk measures is also law-invariant in the sense that two random variables with the same distribution have the same risk.

With a coherent risk measure, one can formulate the corresponding mean-risk portfolio optimization problem by replacing variance with the risk measure. For average value-at-risk, this problem is considered in \cite{rockafellar} in the form of risk minimization subject to an expected return constraint. When jointly Gaussian asset returns are assumed, the risk objective function reduces to the sum of a linear function and the square-root of a quadratic form. The special structure of this case is exploited in \cite{landsmanold}, \cite{owadally}, where analytical results are obtained. A more general objective function in which a differentiable function of variance is added to a linear function is considered in \cite{landsman}, which also provides closed-form solutions. It should be noted that all of these works assume linear constraints.

In this paper, we consider a ``risk-risk problem" where a coherent risk measure is minimized subject to a constraint on a second coherent risk measure. The purpose of using two risk measures is to take into account two risk perceptions when choosing a portfolio. The principle risk measure to be minimized may reflect the risk perception of the portfolio manager while the secondary risk measure in the constraint reflects that of an external authority. Similar to the mean-variance case, the single-objective problem we consider can be seen as a scalarization of a biobjective problem where the objectives are the risk measures of the two bodies who are supposed to choose a portfolio jointly. One advantage of our framework is that it includes both versions of the mean-risk problem as special cases: the one that minimizes risk as well as the one that maximizes expected return.

We first study the risk-risk problem in a general setting where the underlying asset returns are in some $L^p$ space with $p\in[1,+\infty]$ and they have an arbitrary joint distribution with possible correlations. Assuming that the two risk measures are continuous from below so that the suprema in the dual representations are attained at some dual probability measures, we derive a simple dual problem with a linear objective and a linear equality constraint in addition to domain constraints for the dual variables. As the main result of Section~\ref{arbitrary}, under certain constraint qualifications, we show that an optimal solution for the portfolio optimization problem can be obtained as the Lagrange multiplier of the equality constraint of the dual problem at optimality.

At the technical level, the risk-risk problem is a finite-dimensional convex optimization problem. We use the standard Slater's condition (see Assumption~\ref{slater} below) as a constraint qualification to guarantee the existence of optimal Lagrange multipliers for the constraints. Then, we work on the Lagrange dual problem and refine it further by introducing additional dual variables through the dual representations of the two coherent risk measures. This refinement yields a finalized equality-constrained dual problem which is infinite-dimensional due to the dual densities related to the risk measures. To guarantee the existence of an optimal Lagrange multiplier attached to the equality constraint, we need a second constraint qualification. However, the usual Slater's condition with interiority assumptions for the domain constraints is not suitable for this setting due to the fact that many sets in $L^q$ (with $\frac1p+\frac1q=1$), even the positive cone $L^q_+$, may fail to have an empty interior. For this reason, we use the notion of quasi relative interior and the related mild constraint qualification in \cite{borwein} (see Assumption~\ref{quasirelint} below), which still guarantees the existence of an optimal Lagrange multiplier.

In Section~\ref{continuousproblem}, we study the special case where the asset returns are jointly Gaussian and the risk measures are law-invariant. By exploiting the properties of Gaussian distribution and those of the risk measures, the portfolio optimization problem reduces to a problem only with the square-root of a quadratic function and a linear function in the objective and in the constraints. In particular, unlike the above-mentioned works for the Gaussian case, we have a nonlinear constraint that imposes an upper bound on the sum of the square-root of a quadratic form and a linear function.

In the Gaussian case, we observe that the problem can be solved with the help of the hyperbola appearing in the analysis of the mean-variance problem as in \cite{merton}. Indeed, as an associated problem, we consider the minimization of a linear function subject to a linear constraint \emph{over this hyperbola}, which is simply a two-dimensional problem and has a clear geometric interpretation. Using this problem, we provide a complete analysis of the main problem. In particular, we identify all cases in which an optimal solution exists, a unique optimal solution exists, the infimum is finite but not attained, and the problem is unbounded (Section~\ref{continuousproblem}). We provide closed-form expressions for an optimal portfolio, whenever it exists.

\section{Mathematical setup}\label{setup}

\subsection{Portfolios}

We are concerned with various risk-averse versions of the portfolio selection problem on a domain of finitely many risky assets with possibly correlated returns in a one-period market model. To introduce the setup of the problem, we let $n\geq 2$ be an integer denoting the number of assets in the market and write $\N=\cb{1,\ldots,n}$ for the set of these assets. As usual, we denote by $\R^n$ the $n$-dimensional real Euclidean space and by $\R^n_+$ the cone of all vectors $x=(x_1,\ldots,x_n)^{\mathsf{T}}\in\R^n$ with $x_i\geq 0$ for each $i\in\cb{1,\ldots,n}$. For $x,z\in\R^n$, we define their scalar product by $x^{\mathsf{T}}z\coloneqq\sum_{i=1}^nx_iz_i$. Let us fix a probability space $(\O,\F,\Pr)$. We denote by $L_n^0$ the space of all $n$-dimensional random vectors distinguished up to almost sure equality. For each $p\in [1,+\infty)$, we define $L^p_n=\cb{X\in L_n^0\mid \E\sqb{\abs{X}^p}}<+\infty$, and for $p=+\infty$, we define $L_n^\infty=\cb{X\in L_n^0\mid \exists c>0\colon \Pr\cb{\abs{X}\leq c}=1}$, where $\abs{\cdot}$ is an arbitrary norm on $\R^n$. We write $L^p=L_1^p$ for each $p\in\cb{0}\cup[1,+\infty]$.

Let us fix $p\in[1,+\infty]$ and consider a possibly correlated random vector $X=(X_1,\ldots,X_n)^{\mathsf{T}}\in L_n^p$. For each $i\in\N$, the random variable $X_i$ denotes the return of the $i^{\text{th}}$ asset for a fixed period as a multiple of the initial price of that asset. In our context, a portfolio is defined as a vector in $\R^n$ each of whose components denotes the weight of the corresponding asset in the portfolio based on the asset prices at the beginning of the period. Hence, the set of all portfolios is the set
\begin{equation}\label{defnw}
\W\coloneqq\cb{w\in\R^n\mid \sum_{i=1}^n w_i= 1}=\cb{w\in\R^n\mid \1^{\mathsf{T}}w = 1}.
\end{equation}
When shortselling is not allowed, we will restrict ourselves to the portfolios in the subset
\begin{equation}\label{posw}
\W_+\coloneqq\W\cap \R^n_+,
\end{equation}
which is the $(n-1)$-dimensional unit simplex. Note that, for a portfolio $w\in\W$, we have $w^{\mathsf{T}}X\in L^p$, which denotes the return of the porfolio.

\subsection{Risk measures}\label{rm}

We provide a quick review of the theory of risk measures on $L^p$ with $p\in[1,+\infty]$. The reader is referred to \cite{kaina} (for $p\in[1,+\infty)$) and to \citet{fs:sf} (for $p=+\infty$) for a detailed account of the convex-analytic properties of risk measures on $L^p$.

For $Y_1,Y_2\in L^p$, we write $Y_1\leq Y_2$ if $Y_1(\o)\leq Y_2(\o)$ for $\Pr$-almost every $\o\in\O$ and $Y_1\sim Y_2$ if $Y_1$ and $Y_2$ are identically distributed. A functional $\rho\colon L^p\to\R\cup\{+\infty\}$ is said to be a coherent risk measure if it satisfies the following properties.
\begin{enumerate}[\bf (i)]
	\item \textbf{Monotonicity:} $Y_1 \leq Y_2$ implies $ \rho(Y_1) \geq \rho(Y_2)$ for every $Y_1,Y_2\in L^p$.
	\item \textbf{Translativity:} It holds $\rho(Y + y)=\rho(Y)-y$ for every $Y\in L^p$ and $y\in\R$.
	\item \textbf{Subadditivity:} It holds $\rho(Y_1 + Y_2) \leq\rho(Y_1) +\rho(Y_2)$ for every $Y_1,Y_2\in L^p$.
	\item \textbf{Positive homogeneity:} It holds $\rho(\lambda Y)=\lambda \rho(Y)$ for every $Y\in L^p$ and $\lambda \geq 0$.
\end{enumerate}
Clearly, positive homogeneity implies the following property.
\begin{enumerate}[\bf (i)]
	\setcounter{enumi}{4}
	\item \textbf{Normalization:} It holds $\rho(0)=0$.
\end{enumerate}
Moreover, it is easy to check that, under positive homogeneity, subadditivity is equivalent to the following property.
\begin{enumerate}[\bf (i)]
	\setcounter{enumi}{5}
	\item \textbf{Convexity:} It holds $\rho(\lambda Y_1+(1-\lambda)Y_2)\leq \lambda \rho(Y_1)+(1-\lambda)\rho(Y_2)$ for every $Y_1,Y_2\in\Y$ and $\lambda\in[0,1]$.
\end{enumerate}

Let $\rho$ be a coherent risk measure. In Section~\ref{arbitrary}, we need the following additional property:
\begin{enumerate}[\bf (i)]
	\setcounter{enumi}{6}
	\item \textbf{Finiteness:} $\rho(Y)<+\infty$ for every $Y\in L^p$.
\end{enumerate}

Let $\mathcal{M}_1(\Pr)$ denote the set of all probability measures on $(\O,\F)$ that are absolutely continuous with respect to $\Pr$. Let $q\in[1,+\infty]$ such that $\frac1p+\frac1q=1$ and define
\[
\mathcal{M}_1^{q}(\Pr)\coloneqq\cb{\Q\in\mathcal{M}_1(\Pr)\mid \frac{d\Q}{d\Pr}\in L^q}.
\]
Note that $\mathcal{M}_1^{q}(\Pr)=\mathcal{M}_1(\Pr)$.

If $p\in [1,+\infty)$, then finiteness is equivalent to having a dual representation of the form
\begin{equation}\label{dualrep}
\rho(Y)=\max_{\Q\in \mathcal{Q}}\E^{\Q}\sqb{-Y},\quad Y\in L^p,
\end{equation}
for some convex set $\mathcal{Q}\subseteq \mathcal{M}_1^q(\Pr)$ of probability measures such that the corresponding set
\[
\mathcal{D}(\mathcal{Q})\coloneqq \cb{\frac{d\Q}{d\Pr}\mid \Q\in\mathcal{Q}}
\]
of Radon-Nikodym derivatives is $\sigma(L^q,L^p)$-compact; see \citet[Theorem~2.11]{kaina}. Moreover, finiteness also implies that $\rho$ satisfies the following property \cite[Theorem~3.1]{kaina}: 
\begin{enumerate}[\bf (i)]
	\setcounter{enumi}{7}
	\item \textbf{Continuity from below:} If $Y, Y_1,Y_2,\ldots\in L^p$ such that $Y_1\leq Y_2\leq\ldots$ and $\lim_{k\rightarrow\infty}Y_k=Y$ $\Pr$-almost surely, then $\lim_{k\rightarrow\infty}\rho(Y_k)=\rho(Y)$.
\end{enumerate}

If $p=+\infty$, then monotonicity and translativity ensure that $\rho(Y)<+\infty$ for every $Y\in L^\infty$ without an additional assumption. Nevertheless, if one further assumes continuity from below, then a representation of the form \eqref{dualrep} holds for some convex set $\mathcal{Q}\subseteq\mathcal{M}_1(\Pr)$ such that $\mathcal{D}(\mathcal{Q})$ is $\sigma(L^1,L^\infty)$-compact; see \citet[Theorem~3.6]{kaina}.

Finally, we formulate the following additional property that will be needed in Section~\ref{continuousproblem}.
\begin{enumerate}[\bf (i)]
	\setcounter{enumi}{8}
	\item \textbf{Law-invariance:} $Y_1\sim Y_2$ implies $\rho(Y_1)=\rho(Y_2)$ for every $Y_1,Y_2\in L^p$.
\end{enumerate}

Let us recall three commonly used risk measures: negative expected value, value-at-risk and average value-at-risk.

\begin{example}\label{expectation}
	(Negative expected value) Let $p=1$ and take $\rho(Y)=\E\sqb{-Y}$ for every $Y\in L^1$. It is easy to check that $\rho$ satisfies properties (i)-(viii) above. In the dual representation \eqref{dualrep}, we simply have $\mathcal{Q}=\cb{\Pr}$ so that $\mathcal{D}(\mathcal{Q})=\cb{1}\subseteq L^\infty$.
	\end{example}

\begin{example}\label{var}
	(Value-at-risk) Let $p=1$. Let $\theta\in(0,1)$ be a probability level. The \emph{value-at-risk} at level $\theta$ for a random variable $Y\in L^1$ is defined as
	\[
	V@R_\theta(Y)\coloneqq \sup\cb{r\in\R\mid \Pr\cb{Y+r\leq 0}>\theta}.
	\]
	It is well-known that $V@R_\theta$ is a law-invariant positively homogeneous risk measure which fails to be convex. However, if $X$ is a Gaussian random vector and $\Y$ is the Gaussian subspace of $L^2$ spanned by $X_1,\ldots,X_n$ and the constant random variable $1$, it holds (see Proposition~\ref{reformula} below)
	\[
	V@R_\theta(Y)=\Phi^{-1}(1-\theta)\sqrt{\V(Y)}-\E\sqb{Y}
	\]
	for every $Y\in\Y$, where $\Phi$ is the cumulative distribution function of a standard Gaussian random variable $Z$. In particular, $V@R_\theta(Z)=\Phi^{-1}(1-\theta)$. As $Y\mapsto \sqrt{\V(Y)}$ is a convex function on $L^2$, $V@R_\theta$ is a law-invariant coherent risk measure on $\Y$.
\end{example}

\begin{example}\label{cvar}
	(Average value-at-risk) Let $\theta\in(0,1)$ be a probability level. The \emph{average value-at-risk} at level $\theta$ for $Y\in L^1$ is defined as
	\[
	AV@R_\theta (Y)\coloneqq \frac{1}{\theta}\int_0^\theta V@R_u(Y)du.
	\]
	It is well-known that $AV@R_\theta$ is a law-invariant coherent risk measure on $L^1$. In the dual representation in \eqref{dualrep}, we may take $\mathcal{Q}=\{\Q\in\mathcal{M}_1(\Pr)\mid \Pr\{\frac{d\Q}{d\Pr}\leq \frac{1}{\theta}\}=1\}$ so that
	\[
	\mathcal{D}(\mathcal{Q})=\cb{V\in L^\infty\mid \Pr\cb{0\leq V\leq \frac{1}{\theta}}=1}.
	\]
	 On the other hand, for every $Y\in\Y$, where $\Y$ is the Gaussian subspace $\Y$ of $L^2$ in Example~\ref{var}, we have
	\[
	AV@R_\theta(Y)=AV@R_\theta(Z)\sqrt{\V(Y)}-\E\sqb{Y},
	\]
	where $Z$ is a standard Gaussian random variable with
	\[
	AV@R_\theta(Z)=\int_0^\theta \Phi^{-1}(1-u)du.
	\]
\end{example}

\subsection{The portfolio optimization problem}

In this section, we formulate the continuous portfolio optimization problem of our interest.

To model risk-aversion, let $\rho_1,\rho_2\colon L^p\to\R$ be two arbitrary coherent risk measures. The aim of the portfolio manager is to choose a portfolio $w\in\W$ that minimizes the type 1 risk $\rho_1(w^{\mathsf{T}}X)$ while controlling the type 2 risk $\rho_2(w^{\mathsf{T}}X)$ within a fixed threshold level $r\in\R$, that is, while satisfying
\[
\rho_2(w^{\mathsf{T}}X)\leq r,
\]
which we refer to as the \emph{risk constraint}. The use of two risk measures makes sense in cases where the portfolio manager has the right to choose the portfolio using $\rho_1$ as the suitable risk measure for her risk perception but an external regulatory authority with a different risk perception reflected by $\rho_2$ imposes the risk constraint as an obligation for the portfolio manager. It also makes sense when the portfolio manager wishes to work with two risk measures, the principle one ($\rho_1$) having a higher seniority than the other ($\rho_2$). In particular, this framework covers as special cases the problem of maximizing expected return subject to a risk constraint if we take $\rho_1(Y)=\E\sqb{-Y}$ for each $Y\in L^p$, as well as the problem of minimizing (the type 1) risk while maintaining a high-enough expected return if we take $\rho_2(Y)=\E\sqb{-Y}$ for each $Y\in L^p$.

With these risk considerations, we formulate the continuous portfolio optimization problem with shortselling as
\begin{align*}
&\text{minimize}\; \;\;                  \rho_1(w^{\mathsf{T}}X)\tag{$\P(r)$}\\
&\text{subject to}\;\;                   \rho_2(w^{\mathsf{T}}X)\leq r \\
& \quad\quad\quad\quad\;\;\;      w\in \W.
\end{align*}
In this paper, we provide analytical results for $(\P(r))$ in two cases:
\begin{itemize}
\item \textbf{General case:} For a random return vector $X$ with an arbitrary distribution and assuming that $\rho_1, \rho_2$ are continuous from below, we characterize an optimal solution for $(\P(r))$ as a Lagrange multiplier of an associated dual problem in Section~\ref{arbitrary}.
\item \textbf{Gaussian case:} For a Gaussian random return vector $X$ and assuming that $\rho_1,\rho_2$ are law-invariant, we provide a complete analysis of the problem with explicit formulae for an optimal solution and identify the cases where it exists and where it is unique in Section~\ref{continuousproblem}.
\end{itemize}

\section{The portfolio optimization problem under an arbitrary joint distribution}\label{arbitrary}

In this section, we assume that $X\in L_n^p$ for a fixed $p\in[1,+\infty]$ and $\rho_1,\rho_2$ are arbitrary coherent risk measures on $L^p$ that are finite and continuous from below. In particular, $\rho_1,\rho_2$ are continuous on $L^p$; see \citet[Corollary~2.3]{kaina}. Recalling \eqref{dualrep}, $\rho_1,\rho_2$ admit dual representations of the form
\[
\rho_1(Y)=\max_{\Q_1\in\mathcal{Q}_1}\E^{\Q_1}\sqb{-Y},\quad \rho_2(Y)=\max_{\Q_2\in\mathcal{Q}_2}\E^{\Q_2}\sqb{-Y},
\]
for each $Y\in L^p$, where $\mathcal{Q}_1,\mathcal{Q}_2$ are convex subsets of $\mathcal{M}^q_1(\Pr)$ such that the corresponding density sets $\mathcal{D}(\mathcal{Q}_1), \mathcal{D}(\mathcal{Q}_2)$ are convex $\sigma(L^q,L^p)$-compact subsets of $L^q$. For each $j\in\cb{1,2}$, let us define the continuous convex function $g_{j}\colon\R^n\to\R$ by
\[
g_j(w)= \rho_j(w^{\mathsf{T}}X)=\max_{V\in \mathcal{D}(\mathcal{Q}_j)}\E\sqb{-Vw^{\mathsf{T}}X}
\]
for each $w\in\R^n$.

As a preparation for the statement and the proof of the main result, we recall a few notions and facts from convex analysis. Let $\X$ be an Hausdorff locally convex topological linear space with topological dual $\Y$ and bilinear duality mapping $\ip{\cdot,\cdot}\colon\Y\times\X\to\R$. For the purposes of this paper, we are interested in three special cases:
\begin{enumerate}[(i)]
	\item $\X=\R^n$ with the usual topology, which yields $\Y=\R^n$ together with $\ip{y,x}=y^{\mathsf{T}}x$ for every $x\in\R^n, y\in\R^n$.
	\item $\X=L^q$ with $q\in [1,+\infty)$ with the weak topology $\sigma(L^q,L^p)$, which yields $\Y=L^p$ together with $\ip{Y,U}=\E\sqb{UY}$ for every $U\in L^q$, $Y\in L^p$.
	\item $\X=L^\infty$ with the weak topology $\sigma(L^\infty,L^1)$, which yields $\Y=L^1$ together with $\ip{Y,U}=\E\sqb{UY}$ for every $U\in L^\infty$, $Y\in L^1$.
	\end{enumerate}
Let $A\subseteq \X$ be a set. $\cone(A)\coloneqq \cb{\lambda x\mid \lambda\geq 0, x\in A}$ is called the \emph{conic hull} of $A$. If $A$ is convex, then $\cone(A)$ is a convex cone. For $x\in A$, the convex cone
\[
\mathcal{N}_A(x)\coloneqq\cb{y\in\Y\mid \forall x^\prime\in A\colon \ip{y,x}\geq \ip{y,x^\prime}}
\]
is called the \emph{normal cone} of $A$ at $x$. The function $I_A\colon\X\to\R\cup\cb{+\infty}$ defined by $I_A(x)=0$ for $x\in A$ and $I_A(x)=+\infty$ for $x\in \X\sm A$ is called the \emph{indicator function} of $A$. Note that $A$ is convex if and only if $I_A$ is convex. Let $g\colon \X\to\R\cup\cb{+\infty}$ be a function. For a point $x\in \X$, the set $\partial g(x)\coloneqq\cb{y\in\Y\mid \forall x^\prime\in \X\colon g(x^\prime)\geq g(x)+\ip{y,x^\prime-x}}$ is called the \emph{subdifferential} of $g$ at $x$. If $A$ is a nonempty convex set, then it is well-known that \citep[Section~2.4]{zalinescu} $\partial I_A(x)=\N_A(x)$ for every $x\in A$, and $\partial I_A(x)=\emptyset$ for every $x\in\X\sm A$. The function $g^\ast\colon\Y\to\R\cup\cb{\pm\infty}$ defined by $g^\ast(y)\coloneqq\sup_{x\in\X}\of{\ip{y,x}-g(x)}$ for each $y\in\Y$ is called the \emph{conjugate function} of $g$. Note that $y\in\partial g(x)$ holds if and only if $x\in\partial g^{\ast}(y)$ holds for every $x\in\X,y\in\Y$ such that $g$ is lower semicontinuous at $x$.

To formulate a second constraint qualification, we also need the following. For $A\subseteq\X$, the set
\[
\qri(A)\coloneqq \cb{x\in A\mid \N_A(x)\text{ is a subspace of }\Y}
\]
is called the \emph{quasi relative interior} of $A$ \citep[Proposition~2.8]{borwein}. When $\X=\R^n$, $\qri(A)$ coincides with the relative interior of $A$. In this case, $\qri(A)\neq \emptyset$ whenever $A$ is nonempty, closed and convex. When $\X=L^q$ ($q\in[1,+\infty]$) is considered with the topology $\sigma(L^q, L^p)$ and $A$ is nonempty, closed and convex, one has $\qri(A)\neq\emptyset$ again thanks to \citet[Theorem~2.19]{borwein}. In particular, if $A=L^q_+\coloneqq\cb{U\in L^q\mid \Pr\cb{U\geq 0}=1}$, then $\qri(A)=\cb{U\in L^q\mid \Pr\cb{U>0}=1}$ by \citet[Example~3.11]{borwein}. (For $q<+\infty$, considering the strong and weak topologies on $L^q$ yield the same quasi relative interior for a convex set by \citet[Proposition~2.6]{borwein}.)

To be able to study a dual problem with zero duality gap, we work under the following constraint qualification for $(\P(r))$.

\begin{assumption}\label{slater}
	(Slater's condition) There exists $w\in\W$ such that $\rho_2(w^{\mathsf{T}}X)<r$.
\end{assumption}

The main theorem of this section is Theorem~\ref{dualitythm} below. In its proof, by constructing a Lagrange dual problem for $(\P(r))$ and exploiting the dual representations of $\rho_1,\rho_2$, we obtain the following finalized dual problem $(\D(r))$ with an equality constraint.
\begin{align*}
&\text{maximize}\; \;               -r\E\sqb{M} -\lambda\tag{$\D(r)$}\\
&\text{subject to}\;\;                   \E\sqb{UX}+\E\sqb{MX}-\lambda\1=0\\
& \quad\quad\quad\quad\;\;\;      U\in\mathcal{D}(\mathcal{Q}_1),\; M\in\cone(\mathcal{D}(\mathcal{Q}_2)),\; \lambda\in\R.
\end{align*}
The theorem states that an optimal solution for $(\P(r))$ can be calculated as the Lagrange multiplier of the equality constraint of $(\D(r))$ at optimality.

In addition to Assumption~\ref{slater} for $(\P(r))$, we use in Theorem~\ref{dualitythm} the following constraint qualification for $(\D(r))$ based on quasi relative interior.

\begin{assumption}\label{quasirelint}
	\cite[Corollary~4.8]{borwein} There exist $U\in\qri(\mathcal{D}(\mathcal{Q}_1))$, $M\in\qri(\cone(\mathcal{D}(\mathcal{Q}_2)))$, $\lambda\in\R$ such that
	\[
	\E\sqb{UX}+\E\sqb{MX}-\lambda\1=0.
	\]
	\end{assumption}

Note that Assumption~\ref{quasirelint} simply states that one can find $U\in\qri(\mathcal{D}(\mathcal{Q}_1))$ and $M\in\qri(\cone(\mathcal{D}(\mathcal{Q}_2)))$ such that $\E\sqb{UX}+\E\sqb{MX}$ is a constant vector in $\R^n$. The comparison of this assumption with the usual interior-based constraint qualifications is discussed in Remark~\ref{dualslaterremark} after the examples. This remark is followed by Remark \ref{solvability}, where we comment on the usefulness of Theorem~\ref{dualitythm} for computations.

\begin{theorem}\label{dualitythm}
	Under Assumption~\ref{slater} and Assumption~\ref{quasirelint}, suppose that there exists an optimal solution $(U^\ast,M^\ast, \lambda^\ast)\in L^q\times L^q\times \R$ of $\mathscr{D}(r)$. Then, there exists an optimal Lagrange multiplier $w^\ast\in\R^n$ associated to the equality constraint of $\mathscr{D}(r)$. Moreover, every $w^\ast\in\R^n$ that is the Lagrange multiplier of the equality constraint of $\mathscr{D}(r)$ at optimality is an optimal solution for $(\P(r))$.
	\end{theorem}

We use the following lemma for the proof of Theorem~\ref{dualitythm}, which should be known. For completeness, we present its short proof.

\begin{lemma}\label{subdifflemma}
	Let $w\in\R^n$, $j\in\cb{1,2}$, and define the attainment set
	\begin{equation}\label{scriptv}
	\mathscr{V}_j(w)\coloneqq \argmax_{V\in\mathcal{D}(\mathcal{Q}_j)} \E\sqb{-VX^{\mathsf{T}}w}.
	\end{equation}
	Then, one has
	\[
	\partial g_j(w)= \cb{\E\sqb{-VX}\mid V\in \mathscr{V}_j(w)}.
	\]
\end{lemma}

\begin{proof}
	Since $\rho_j$ is continuous from below, $\mathscr{V}_j(w)\neq \emptyset$. Note that the linear operator $A\colon\R^n\to L^p$ defined by $Aw^\prime\coloneqq X^{\mathsf{T}}w^\prime$ for $w^\prime\in\R^n$ has the adjoint operator $A^\ast\colon L^q\to \R^n$ given by $A^\ast V= \E\sqb{VX}$ for $V\in L^q$. (We consider the $\sigma(L^\infty,L^1)$ topology on $L^\infty$ when $p=+\infty$ so that the dual space of $L^\infty$ is $L^1$.) Let $w\in\R^n$. Since we have $g_j=\rho_j\circ A$ and $\rho_j$ is continuous on $L^p$, by subdifferential calculus rules, e.g.,  \citet[Theorem~2.8.3(iii)]{zalinescu}, 
	\[
	\partial g_j(w)= \cb{A^\ast V\mid V\in \partial \rho_j(Aw)}=\cb{\E\sqb{VX}\mid V\in \partial \rho_j(X^{\mathsf{T}}w)}.
	\]
	On the other hand, since $\rho_j$ is continuous, for each $Y\in L^p$ and $V\in L^q$, we have $V\in \partial \rho_j(Y)$ if and only if $Y\in \partial \rho_j^\ast(V)$. On the other hand, for each $V\in L^q$,
	\begin{align*}
	\rho_j^\ast(V)&=\sup_{Y^\prime\in L^p}\of{\E\sqb{VY^\prime}-\rho_j(Y^\prime)}=\sup_{Y^\prime\in L^p}\of{\E\sqb{VY^\prime}+\inf_{V^\prime\in \mathcal{D}(\mathcal{Q}_j)}\E\sqb{V^\prime Y^\prime}}\\
	&=\inf_{V^\prime\in \mathcal{D}(\mathcal{Q}_j)}\sup_{Y^\prime\in L^p}\E\sqb{(V+V^\prime)Y^\prime}=\inf_{V^\prime\in \mathcal{D}(\mathcal{Q}_j)}I_{\cb{-V^\prime}}(V)= I_{-\mathcal{D}(\mathcal{Q}_j)}(V),
	\end{align*}
	where we use the minimax theorem \cite[Corollary~3.3]{sion} for the third equality thanks to the fact that $\mathcal{D}(\mathcal{Q}_j)$ is a convex $\sigma(L^q,L^p)$-compact set. As a result,
	\begin{align*}
	\partial\rho_j^\ast(V)=\mathcal{N}_{-\mathcal{D}(\mathcal{Q}_j)}(V)	
	=\cb{Y\in L^p\mid \forall V^\prime\in\mathcal{D}(\mathcal{Q}_j)\colon \E\sqb{VY}\geq \E\sqb{-V^\prime Y}}
	\end{align*}
	if $V\in -\mathcal{D}(\mathcal{Q}_j)$ and $\partial\rho_j^\ast(V)=\emptyset$ if $V\in L^q\sm -\mathcal{D}(\mathcal{Q}_j)$.
	Consequently,
	\begin{align*}
	\partial g_j(w)&=\cb{\E\sqb{VX}\mid V\in \partial \rho_j(X^{\mathsf{T}}w)}\\
	&=\cb{\E\sqb{VX}\mid V\in -\mathcal{D}(\mathcal{Q}_j), \ X^\mathsf{T}w\in\partial\rho_j^\ast(V)}\\
	&=\cb{\E\sqb{VX}\mid V\in -\mathcal{D}(\mathcal{Q}_j),\forall V^\prime\in\mathcal{D}(\mathcal{Q}_j)\colon \E\sqb{VX^\mathsf{T}w}\geq \E\sqb{-V^\prime X^\mathsf{T}w}}\\
	&=\cb{\E\sqb{-VX}\mid V\in \mathcal{D}(\mathcal{Q}_j),\forall V^\prime\in\mathcal{D}(\mathcal{Q}_j)\colon \E\sqb{-VX^\mathsf{T}w}\geq \E\sqb{-V^\prime X^\mathsf{T}w}}\\
	&=\cb{\E\sqb{-VX}\mid V\in \mathscr{V}_j(w)}
	\end{align*}
	so that the result follows.
\end{proof}

\begin{proof}[Proof of Theorem~\ref{dualitythm}]
Let us denote by $p$ the optimal value of $(\P(r))$. Thanks to Assumption~\ref{slater}, by strong duality for convex optimization (for instance, by \citet[Theorem~2.9.3]{zalinescu}), $p$ is equal to the optimal value of the corresponding Lagrange dual problem, that is,
\begin{equation}\label{ldual}
p=\sup_{\nu\geq 0,\lambda\in\R}d(\nu,\lambda),
\end{equation}
where, for each $\nu\geq0,\lambda\in\R$,
\[
d(\nu,\lambda)\coloneqq \inf_{w\in\R^n}\of{\rho_1(w^{\mathsf{T}}X)+\nu\of{\rho_2(w^{\mathsf{T}}X)-r}+\lambda\of{\1^{\mathsf{T}}w-1}}.
\]

Let us fix $\nu\geq 0, \lambda\in\R$. Using the dual representations of $\rho_1,\rho_2$,
\begin{align*}
d(\nu,\lambda)
=\inf_{w\in\R^n}\of{\max_{U\in\mathcal{D}(\mathcal{Q}_1)}\E\sqb{-Uw^{\mathsf{T}}X}+\nu\max_{V\in\mathcal{D}(\mathcal{Q}_2)}\E\sqb{-Vw^{\mathsf{T}}X}+\lambda\1^{\mathsf{T}}w}-r\nu -\lambda.
\end{align*}
Let $f(w,U,V)\coloneqq \E\sqb{-Uw^{\mathsf{T}}X}+\nu\E\sqb{-Vw^{\mathsf{T}}X}+\lambda\1^{\mathsf{T}}w$ for each $w\in\R^n,U\in\mathcal{D}(\mathcal{Q}_1),V\in\mathcal{D}(\mathcal{Q}_2)$. Note that $w\mapsto f(w,U,V)$ is convex (affine) and continuous, $(U,V)\mapsto f(w,U,V)$ is concave (affine) and $\sigma(L^q,L^p)$-continuous (continuous), and $\mathcal{D}(\mathcal{Q}_1)\times\mathcal{D}(\mathcal{Q}_2)$ is convex and $\sigma(L^q,L^p)$-compact. Hence, the classical minimax theorem \cite[Corollary~3.3]{sion} ensures that
\begin{align*}
d(\nu,\lambda)&=\sup_{(U,V)\in\mathcal{D}(\mathcal{Q}_1)\times\mathcal{D}(\mathcal{Q}_2)}\inf_{w\in\R^n}\of{\E\sqb{-Uw^{\mathsf{T}}X}+\nu\E\sqb{-Vw^{\mathsf{T}}X}+\lambda\1^{\mathsf{T}}w}-r\nu -\lambda \\
&=\sup_{(U,V)\in\mathcal{D}(\mathcal{Q}_1)\times\mathcal{D}(\mathcal{Q}_2)}\inf_{w\in\R^n}\of{\E\sqb{-UX}+\nu\E\sqb{-VX}+\lambda\1}^{\mathsf{T}}w-r\nu -\lambda.
\end{align*}
Clearly, for every $(U,V)\in\mathcal{D}(\mathcal{Q}_1)\times\mathcal{D}(\mathcal{Q}_2)$,
\begin{equation}\label{infeq}
\inf_{w\in\R^n}\of{\E\sqb{-UX}+\nu\E\sqb{-VX}+\lambda\1}^{\mathsf{T}}w=\begin{cases}0&\text{if }\E\sqb{-UX}+\nu\E\sqb{-VX}+\lambda\1=0,\\ -\infty&\text{else}.\end{cases}
\end{equation}
It follows that
\[
d(\nu,\lambda)=\begin{cases}- r\nu -\lambda&\text{if }\exists(U,V)\in\mathcal{D}(\mathcal{Q}_1)\times\mathcal{D}(\mathcal{Q}_2)\colon \E\sqb{-UX}+\nu\E\sqb{-VX}+\lambda\1=0,\\ -\infty&\text{else}.\end{cases}
\]
So the Lagrange dual problem in \eqref{ldual} takes the more explicit form
\begin{align*}
&\text{maximize}\; \;\;                  -r\nu -\lambda\tag{$\tilde{\D}(r)$}\\
&\text{subject to}\;\;                   \E\sqb{UX}+\nu\E\sqb{VX}-\lambda\1 =0\\
& \quad\quad\quad\quad\;\;\;      U\in\mathcal{D}(\mathcal{Q}_1),\; V\in\mathcal{D}(\mathcal{Q}_2),\; \nu\geq 0,\; \lambda\in\R.
\end{align*}
To avoid the multiplication of the variables $\nu,V$, we make the following change of variables: Note that if $M\in\cone(\mathcal{D}(\mathcal{Q}_2))$, then there exist $\nu\geq 0$ and $V\in\mathcal{D}(\mathcal{Q}_2)$ such that $M=\nu V$: we simply take $\nu=\E\sqb{M}$, and $V=\frac{M}{\nu}$ if $\nu>0$ and an arbitrary $V\in\mathcal{D}(\mathcal{Q}_2)$ if $\nu=0$. Conversely, if $\nu\geq 0$ and $V\in\mathcal{D}(\mathcal{Q}_2)$, then $M=\nu V\in\cone(\mathcal{D}(\mathcal{Q}_2))$. These observations allow us to reformulate $(\tilde{\D}(r))$ as $(\D(r))$. Note that both problems have $p$ as their optimal value.

Let $(U^\ast,M^\ast, \lambda^\ast)\in L^q\times L^q\times \R$ be an optimal solution for $\mathscr{D}(r)$. Thanks to Assumption~\ref{quasirelint} and \citet[Corollary~4.8]{borwein}, there is strong duality with the corresponding Lagrange dual problem that relaxes the equality constraint, that is, we have
\begin{align*}
p&=\inf_{w\in\R^n}\sup_{U\in\mathcal{D}(\mathcal{Q}_1),M\in\cone(\mathcal{D}(\mathcal{Q}_2)),\lambda\in\R}\of{-r\E\sqb{M}-\lambda-w^{\mathsf{T}}\of{\E\sqb{UX}+\E\sqb{MX}-\lambda\1}}\\
&=\inf_{w\in\R^n}\sup_{U\in\mathcal{D}(\mathcal{Q}_1),M\in\cone(\mathcal{D}(\mathcal{Q}_2)),\lambda\in\R}\of{-r\E\sqb{M}-\lambda+\E\sqb{-Uw^{\mathsf{T}}X}+\E\sqb{-Mw^{\mathsf{T}}X}+\lambda w^{\mathsf{T}}\1}.
\end{align*}
Moreover, \citet[Corollary~4.8]{borwein} also ensures that there exists an optimal Lagrange multiplier $w^\ast\in\R^n$. By the first-order condition with respect to $U=U^\ast$, we have
\[
0\in -(w^\ast)^{\mathsf{T}}X -\mathcal{N}_{\mathcal{D}(\mathcal{Q}_1)}(U^\ast),
\]
which means that
\[
\E\sqb{-U^\ast(w^\ast)^{\mathsf{T}}X}\geq \E\sqb{-U^\prime(w^\ast)^{\mathsf{T}}X }
\]
for every $U^\prime\in\mathcal{D}(\mathcal{Q}_1)$, that is,
\[
\rho_1((w^\ast)^\mathsf{T}X)=\E\sqb{-U^\ast (w^\ast)^{\mathsf{T}}X}.
\]
We conclude that $U^\ast\in\mathscr{V}_1(w^\ast)$, where $\mathscr{V}_1(w^\ast)$ is defined by \eqref{scriptv}. In particular, by Lemma~\ref{subdifflemma},
\begin{equation}\label{subdiff1}
\E\sqb{-U^\ast X}\in\partial g_1(w^\ast).
\end{equation}

Similarly, the first-order condition with respect to $M=M^\ast$ yields
\[
\E\sqb{-M^\ast\of{(w^\ast)^{\mathsf{T}}X+r}}\geq \E\sqb{-M^\prime\of{(w^\ast)^{\mathsf{T}}X+r} }
\]
for every $M^\prime\in\cone(\mathcal{D}_2)$, that is,
\begin{equation}\label{someeq}
\E\sqb{-M^\ast\of{(w^\ast)^{\mathsf{T}}X+r}}=\max_{M^\prime\in\cone(\mathcal{D}(\mathcal{Q}_2))}\E\sqb{-M^\prime\of{(w^\ast)^{\mathsf{T}}X+r} }.
\end{equation}
Since $\cone(\mathcal{D}_2)$ is a cone, the quantity $\sup_{M^\prime\in\cone(\mathcal{D}(\mathcal{Q}_2))}\E\sqb{-M^\prime((w^\ast)^{\mathsf{T}}X+r)}$ can either take value $0$ or $+\infty$. Since $\E\sqb{-M^\ast((w^\ast)^{\mathsf{T}}X+r)}$ is a finite number, both sides of \eqref{someeq} must be equal to zero. Moreover,
\begin{align*}
0=\max_{M^\prime\in\cone(\mathcal{D}(\mathcal{Q}_2))}\E\sqb{-M^\prime\of{(w^\ast)^{\mathsf{T}}X+r}}&=\of{\sup_{\lambda^\prime\geq 0}\lambda^\prime}\of{\max_{V^\prime\in\mathcal{D}(\mathcal{Q}_2)}\E\sqb{-V^\prime\of{(w^\ast)^{\mathsf{T}}X+r}}}\\
&=+\infty\cdot \rho_2((w^\ast)^{\mathsf{T}}X+r)=+\infty\cdot \of{\rho_2((w^\ast)^{\mathsf{T}}X)-r}.
\end{align*}
Hence, we have $\rho_2((w^\ast)^{\mathsf{T}}X)=r$.

Let $\nu^\ast=\E\sqb{M^\ast}$. Suppose first that $\nu^\ast>0$. Let $V^\ast\coloneqq \frac{M^\ast}{\nu^\ast}\in\mathcal{D}(\mathcal{Q}_2)$. Then,
\[
\E\sqb{-M^\ast((w^\ast)^{\mathsf{T}}X+r)}=\nu^\ast\E\sqb{-V^\ast\of{(w^\ast)^{\mathsf{T}}X+r}}=0
\]
so that $\E\sqb{-V^\ast(w^\ast)^{\mathsf{T}}X}=r$. Hence,
\[
\E\sqb{-V^\ast(w^\ast)^{\mathsf{T}}X}=r=\rho_2((w^\ast)^{\mathsf{T}}X)=\max_{V^\prime\in\mathcal{D}(\mathcal{Q}_2)}\E\sqb{-V^\prime(w^\ast)^{\mathsf{T}}X},
\]
that is, $V^\ast\in\mathscr{V}_2(w^\ast)$. In particular,
\[
\E\sqb{-V^\ast X}\in\partial g_2(w^\ast).
\]
Next, suppose that $\nu^\ast=0$, that is, $M^\ast=0$ $\Pr$-almost surely. Let us pick some $V^\ast\in\mathscr{V}_2(w^\ast)$ arbitrarily. (We know that $\mathscr{V}_2(w^\ast)\neq \emptyset$ since $\rho_2$ is assumed to be continuous from below.) In both cases, we may write $M^\ast=\nu^\ast V^\ast$ and
\begin{equation}\label{subdiff2}
\E\sqb{-M^\ast X}=\nu^\ast\E\sqb{-V^\ast X}\in\nu^\ast\partial g_2(w^\ast).
\end{equation}

By the feasibility of $(U^\ast,M^\ast,\lambda^\ast)$ for $(\D(r))$, 
\begin{equation}\label{suff}
\E\sqb{-U^\ast X}+\E\sqb{-M^\ast X}+\lambda^\ast\1=\E\sqb{-U^\ast X}+\nu^\ast\E\sqb{-V^\ast X}+\lambda^\ast\1 = 0.
\end{equation}
Hence, by \eqref{subdiff1}, \eqref{subdiff2}, \eqref{suff}, we conclude that
\[
0\in \partial g_1(w^\ast)+\nu^\ast \partial g_2(w^\ast)+\lambda^\ast\1.
\]
Finally, by the first-order condition with respect to $\lambda=\lambda^\ast$, we get
\[
\1^\mathsf{T}w^\ast=1,
\]
that is, $w^\ast\in\W$. Therefore, we establish the Karush-Kuhn-Tucker conditions for $(\P(r))$ at $w=w^\ast$. By \citet[Theorem~2.9.3]{zalinescu}, we conclude that $w^\ast$ is an optimal solution for $(\mathscr{P}(r))$.
\end{proof}

\begin{remark}
	Let us comment on the roles of the two constraint qualifications and the assumption about the existence of an optimal solution for $(\D(r))$. In the proof of Theorem~\ref{dualitythm}, note that Assumption~\ref{slater} already guarantees the existence of an optimal solution $(\bar{\nu},\bar{\lambda})$ for the Lagrange dual problem in \eqref{ldual}. Nevertheless, the reformulated problem $(\tilde{\D}(r))$ has two additional variables, $U$ and $V$, which, together with $\nu,\lambda$, combine into $U,M,\lambda$ in the finalized dual problem $(\D(r))$. As a result, the existence of an optimal solution for $(\D(r))$ is not guaranteed \emph{a priori}. Once such an optimal solution is assumed, Assumption~\ref{quasirelint} automatically yields the existence of an optimal Lagrange multiplier for the equality constraint in $(\D(r))$, which is shown to give an optimal solution for the original problem $(\P(r))$.
\end{remark}

In the following examples, we consider a few special choices of the risk measures. We work with $p=1$ in all examples.

\begin{example}\label{example1}
	Let $\rho_1$ be the average value-at-risk at a probability level $\theta\in(0,1)$ (Example~\ref{cvar}) and $\rho_2$ the negative expected value (Example~\ref{expectation}). In this case:
	\begin{align*}
	&\mathcal{D}(\mathcal{Q}_1)=\cb{U\in L^\infty\mid \Pr\cb{0\leq U\leq \frac{1}{\theta}}=1},\quad \qri(\mathcal{D}(\mathcal{Q}_1))=\cb{U\in L^\infty\mid \Pr\cb{0< U< \frac{1}{\theta}}=1},\\
	&\cone(\mathcal{D}(\mathcal{Q}_2))=\cone({1})=\R_+,\quad\quad \quad \quad\quad \quad \quad \qri(\cone(\mathcal{D}(\mathcal{Q}_2)))=(0,+\infty),
		\end{align*}
where the quasi relative interiors can be calculated by following a similar procedure as in \citet[Example~3.11]{borwein}. It is easy to observe that Assumption~\ref{quasirelint} is equivalent to the existence of a probability measure $\Q$ on $(\O,\F)$ that is equivalent to $\Pr$ such that $\frac{d\Q}{d\Pr}\leq \frac{1}{\theta}$ $\Pr$-almost surely and $\E^\Q\sqb{-X}$ is in the conic convex hull of the set $\cb{\E\sqb{X},\1,-\1}$. In particular, if $\E\sqb{X_1}=\ldots=\E\sqb{X_n}$, then this condition is satisfied by $\Q=\Pr$. Moreover, the dual problem $(\D(r))$ becomes
 \begin{align*}
 &\text{maximize}\; \;               -r m -\lambda\\
 &\text{subject to}\;\;                   \E\sqb{UX}+m\E\sqb{X}-\lambda\1=0\\
  & \quad\quad\quad\quad\;\;\;     \E\sqb{U}=1\\
  & \quad\quad\quad\quad\;\;\;     0\leq U\leq\frac{1}{\theta}\; \Pr\text{-almost surely}\\
 & \quad\quad\quad\quad\;\;\;      U\in L^\infty,\; m\geq 0,\; \lambda\in\R,
 \end{align*}
which is a linear programming problem in an infinite-dimensional setting. When $(\O,\F,\Pr)$ is a finite probability space, it reduces to a finite-dimensional linear programming problem.
	\end{example}

\begin{example}\label{example2}
	We switch the roles of negative expected value and average value-at-risk in Example~\ref{example1} so that
	\begin{align*}
	&\mathcal{D}(\mathcal{Q}_1)=\qri(\mathcal{D}(\mathcal{Q}_1))=\cb{1}\subseteq L^\infty,\\
	&\cone(\mathcal{D}(\mathcal{Q}_2))=\cb{M\in L^\infty\mid \Pr\cb{0\leq M\leq \frac{\E\sqb{M}}{\theta}}=1},\\
	& \qri(\cone(\mathcal{D}(\mathcal{Q}_2)))=\cb{M\in L^\infty\mid \Pr\cb{0< M< \frac{\E\sqb{M}}{\theta}}=1}.
	\end{align*}
	In this case, Assumption~\ref{quasirelint} is equivalent to the existence of a finite measure with density $M$ such that $\theta M<\E\sqb{M}$ $\Pr$-almost surely and $\E\sqb{-MX}$ is in the unbounded polyhedral set $\{\E\sqb{X}-\lambda\1\mid \lambda\in\R\}$. In particular, if $\E\sqb{X_1}=\ldots=\E\sqb{X_n}$, then this condition is satisfied by $M\equiv 1$. Moreover, the dual problem $(\D(r))$ becomes
	\begin{align*}
	&\text{maximize}\; \;               -r \E\sqb{M} -\lambda\\
	&\text{subject to}\;\;                   \E\sqb{X}+\E\sqb{MX}-\lambda\1=0\\
	& \quad\quad\quad\quad\;\;\;     0\leq \theta M\leq\E\sqb{M}\; \Pr\text{-almost surely}\\
	& \quad\quad\quad\quad\;\;\;      M\in L^\infty,\; \lambda\in\R,
	\end{align*}
	which reduces to a finite-dimensional linear programming problem when $(\O,\F,\Pr)$ is a finite probability space.
	\end{example}

\begin{remark}\label{dualslaterremark}
	Note that Assumption~\ref{quasirelint} is a constraint qualification for the dual problem $(\D(r))$. A more standard alternative of it would assume the existence of $U\in\interior(\mathcal{D}(\mathcal{Q}_1))$, $M\in\interior(\cone(\mathcal{D}(\mathcal{Q}_2)))$, $\lambda\in\R$ such that $\E\sqb{UX}+\E\sqb{MX}-\lambda\1=0$, where $\interior$ denotes topological interior. However, in infinite-dimensional spaces, many convex sets that show up in applications have empty interior. For instance, it is well-known that, for $q\in [1,+\infty)$, we have $\interior L^q_+=\emptyset$ unless $L^q$ is a finite-dimensional space, that is, the underlying probability space is isomorphic to a finite probability space (see, for instance, \citet[Example~10.1.3]{dynamics} and \citet[Example~2.12]{gluck}). Similarly, we have $\interior(\mathcal{D}(\mathcal{Q}_1))=\emptyset$ in Example~\ref{example1} and $\interior(\cone(\mathcal{D}(\mathcal{Q}_2)))=\emptyset$ in Example~\ref{example2} unless $L^q$ is a finite-dimensional space. Hence, the standard alternative is never satisfied for our examples whenever we deviate from the finite-dimensional case. Hence, Assumption~\ref{quasirelint}, which uses quasi relative interior instead of usual interior, is much weaker than its interior-based counterpart (see, for instance, \citet[Theorem~2.9.6]{zalinescu} for a strong duality theorem with an interior-based constraint qualification).
\end{remark}

\begin{remark}\label{solvability}
	 Let us comment on the benefit of Theorem~\ref{dualitythm} for computing an optimal solution for the primal portfolio optimization problem $(\P(r))$. When $(\O,\F,\Pr)$ is a finite probability space, the dual problem $(\D(r))$ reduces to a finite-dimensional convex optimization problem with the set constraints $U\in \mathcal{D}(\mathcal{Q}_1)$ and $M\in \cone(\mathcal{D}(\mathcal{Q}_2))$. If these constraints can be represented by finitely many (convex) inequalities, then $(\D(r))$ can be solved by commercial software for convex optimization such as CVX, which also return the value of the Lagrange multipliers for the constraints at (approximate) optimality. Hence, without an additional procedure, an optimal solution for $(\P(r))$ is readily computed by the solver as the dual multiplier of the constraint $\E\sqb{UX}+\E\sqb{MX}-\lambda\1=0$ at optimality. As noted in Example~\ref{example1} and Example~\ref{example2}, when the classical coherent risk measures negative expected value and average value-at-risk are used in the problem, the set constraints $U\in \mathcal{D}(\mathcal{Q}_1)$ and $M\in \cone(\mathcal{D}(\mathcal{Q}_2))$ are easily represented by finitely many linear inequalities so that $(\D(r))$ even reduces to a linear programming problem. In this case, many more commerical solvers are available (for instance, CPLEX, Gurobi) and they also return the values of the dual variables associated to the constraints at optimality. Hence, in the cases where $(\D(r))$ reduces to a standard convex/linear optimization problem and it has an optimal solution, thanks to Theorem~\ref{dualitythm}, an optimal portfolio for $(\P(r))$ is returned by commercial solvers. It should also be noted that the idea of recovering a primal optimal solution from the dual multipliers of the dual problem is pretty well-known in the nonsmooth/stochastic optimization literature (see, \cite{barahona}, \cite{larsson}, for instance). Hence, at a high level, Theorem~\ref{dualitythm} can be considered as a result in the same spirit.
	\end{remark}

We finish this section by providing an analogous dual problem and an optimality result for the case where shortselling is not allowed, namely, for the problem
\begin{align*}
&\text{minimize}\; \;\;                  \rho_1(w^{\mathsf{T}}X)\tag{$\P_+(r)$}\\
&\text{subject to}\;\;                   \rho_2(w^{\mathsf{T}}X)\leq r \\
& \quad\quad\quad\quad\;\;\;      w\in \W_+,
\end{align*}
where $\W_+$ is defined by \eqref{posw}. The analysis of $(\P_+(r))$ is very similar to that of $(\P(r))$ and it yields the finalized dual problem
\begin{align*}
&\text{maximize}\; \;               -r\E\sqb{M} -\lambda\tag{$\D_+(r)$}\\
&\text{subject to}\;\;                   \E\sqb{UX}+\E\sqb{MX}-\lambda\1\leq 0\\
& \quad\quad\quad\quad\;\;\;      U\in\mathcal{D}(\mathcal{Q}_1),\; M\in\cone(\mathcal{D}(\mathcal{Q}_2)),\; \lambda\in\R.
\end{align*}
We have the following duality result which works under modified versions of Assumption~\ref{slater} and Assumption~\ref{quasirelint}.
\begin{theorem}\label{dualitythm2}
Assume that there exists $w\in\W_+$ such that $\rho_2(w^{\mathsf{T}}X)< r$ and $w_i>0$ for every $i\in\mathcal{N}$. Assume further that there exist $U\in\qri(\mathcal{D}(\mathcal{Q}_1))$, $M\in\qri(\cone(\mathcal{D}(\mathcal{Q}_2)))$, $\lambda\in\R$ such that $\E\sqb{UX_i}+\E\sqb{MX_i}-\lambda<0$ for every $i\in\mathcal{N}$. Suppose that there exists an optimal solution $(U^\ast,M^\ast, \lambda^\ast)\in L^q\times L^q\times \R$ of $(\mathscr{D}_+(r))$. Then, there exists an optimal Lagrange multiplier $w^\ast\in\R_+^n$ associated to the inequality constraint of $(\mathscr{D}_+(r))$, and $w^\ast$ is an optimal solution for $(\P_+(r))$.
\end{theorem}

\begin{proof}
	The proof goes along the same lines as the proof of Theorem~\ref{dualitythm}. The only important change is in \eqref{infeq}:
	\[
	\inf_{w\in\R_+^n}\of{\E\sqb{-UX}+\nu\E\sqb{-VX}+\lambda\1}^{\mathsf{T}}w=\begin{cases}0&\text{if }\E\sqb{-UX}+\nu\E\sqb{-VX}+\lambda\1\geq 0,\\ -\infty&\text{else},\end{cases}
	\]
	which is the reason for having an inequality constraint in $(\D_+(r))$. The rest follows in a standard manner.
	\end{proof}

\section{The  portfolio optimization problem under the multivariate Gaussian distribution}\label{continuousproblem}

In this section, we study the problem $(\P(r))$ under the special case that $X$ is a Gaussian random vector and $\rho_1, \rho_2$ are law-invariant. Under these assumptions, it turns out that the analysis of $(\P(r))$ can be performed in terms of the hyperbola appearing in the classical Markowitz problem and an optimal solution for $(\P(r))$ can be calculated with an explicit formula whenever it exists. The aim of this section is to provide an analysis that is peculiar to the Gaussian case; hence, we follow a route that is quite different from the general duality-based approach in Section~\ref{arbitrary}.

We assume that $X=(X_1,\ldots,X_n)^{\mathsf{T}}\in L_n^2$ is a Gaussian random vector with mean vector $m=(m_1, \ldots, m_n)^\mathsf{T}$ and covariance matrix $C\in\R^{n\times n}$. We further assume that $m$ and $\1\coloneqq(1,\ldots,1)^{\mathsf{T}}\in\R^n$ are linearly independent and that $C$ is a nonsingular matrix with inverse $C^{-1}$. Hence, $C$ is a symmetric positive definite matrix with strictly positive eigenvalues.

Note that, for a portfolio $w\in\W$, its return $w^{\mathsf{T}}X$ is a Gaussian random variable. A simple calculation yields that the corresponding expected value and variance are given by
\begin{equation}\label{expvar}
\mu_w\coloneqq\E\sqb{w^{\mathsf{T}}X}=m^{\mathsf{T}}w,\quad \sigma^2_w\coloneqq\V(w^{\mathsf{T}}X)=w^{\mathsf{T}}Cw,
\end{equation}
respectively. For every $w\in\W$, we may write
\begin{equation}\label{standardization}
w^{\mathsf{T}}X=\E\sqb{w^{\mathsf{T}}X}+\sqrt{\V(w^{\mathsf{T}}X)}Z
\end{equation}
for some standard Gaussian random variable $Z$ (with zero mean and unit variance). Using this, we provide an explicit expression for the values of a generic law-invariant coherent risk measure $\rho$ next.

\begin{proposition}\label{reformula}
	Let $\rho$ be a coherent, law-invariant and finite risk measure on $L^2$. For every $w\in\W$, it holds
	\[
	\rho(w^{\mathsf{T}}X)=\rho(Z)\sqrt{w^{\mathsf{T}}Cw}-m^{\mathsf{T}}w,
	\]
	where $Z$ is an arbitrary standard Gaussian random variable.
\end{proposition}

\begin{proof}
	Let $w\in\W$. Using \eqref{standardization}, we obtain
	\[
	\rho(w^{\mathsf{T}}X)=\rho\of{\sqrt{\V(w^{\mathsf{T}}X)}Z+\E\sqb{w^{\mathsf{T}}X}}=\rho(Z)\sqrt{\V(w^{\mathsf{T}}X)}-\E\sqb{w^{\mathsf{T}}X}
	\]
	thanks to the translativity and positive homogeneity of $\rho$. Finally, the number $\rho(Z)$ is free of the choice of the standard Gaussian random variable $Z$ thanks to the law-invariance of $\rho$.
\end{proof}

With a slight abuse of notation, we define $\rho_j\coloneqq \rho_j(Z)\geq 0$ for each $j\in\cb{1,2}$, where $Z$ is a generic standard Gaussian random variable. Thanks to Proposition~\ref{reformula}, we may rewrite $(\P(r))$ as 
\begin{align*}
&\text{minimize}\; \;\;                  \rho_1\sqrt{w^{\mathsf{T}}Cw}-m^{\mathsf{T}}w\tag{$\P(r)$}\\
&\text{subject to}\;\;                     \rho_2\sqrt{w^{\mathsf{T}}Cw}-m^{\mathsf{T}}w\leq r \\
& \quad\quad\quad\quad\;\;\;      \1^{\mathsf{T}}w=1\\
& \quad\quad\quad\quad\;\;\;      w\in\R^n.
\end{align*}

In what follows, we provide an analytical solution for $(\P(r))$, whenever it exists, under all possible relationships among the parameters $m, C, r, \rho_1, \rho_2$. To that end, let us introduce the constants
\begin{equation*}
\a\coloneqq m^{\mathsf{T}}C^{-1}m,\quad \b\coloneqq m^{\mathsf{T}}C^{-1}\1=\1^{\mathsf{T}}C^{-1}m,\quad \g\coloneqq \1^{\mathsf{T}}C^{-1}\1,\quad \delta\coloneqq \a\g-\b^2,
\end{equation*}
which also appear in the analysis of the classical Markowitz problem. As a consequence of the positive definiteness of $C$, it is well-known and easy to check that $\a,\g,\d>0$.

\subsection{The Markowitz hyperbola}

The analysis of the $n$-dimensional portfolio optimization problem $(\P(r))$ is based on an associated two-dimensional optimization problem whose decision variables stand for the standard deviation and expected return of a portfolio. Note that every portfolio $w\in\W$ induces a standard deviation-expected return pair $(\sigma_w,\mu_w)\in\R^2$ of $(\M(r))$ through the definitions $\sigma_w=\sqrt{w^{\mathsf{T}}Cw}$, $\mu_w=m^{\mathsf{T}}r$. The structure of the set $\cb{(\sigma_w,\mu_w)\mid w\in\W}$ is very well-known: this set is the convex hull of the right wing of a hyperbola. The precise version of this classical result is recalled in the next lemma.

\begin{lemma}\label{mw}
	Let $\mu\in\R$ and consider the problem of finding the portfolio with minimum variance among all the portfolios with expected return level $\mu$:
	\begin{align*}
		&\text{minimize}\; \;\;                   w^{\mathsf{T}}Cw\tag{$\A(\mu)$}\\
		&\text{subject to}\;\;                     m^{\mathsf{T}}w= \mu \\
		& \quad\quad\quad\quad\;\;\;      \1^{\mathsf{T}}w=1\\
		& \quad\quad\quad\quad\;\;\;      w\in\R^n.
	\end{align*}
	The problem $\A(\mu)$ has a unique optimal solution given by
	\begin{equation}\label{wfml}
	w(\mu)\coloneqq\frac{1}{\d}\of{(\g\mu-\b)C^{-1}m+(\a-\b\mu)C^{-1}\1}
	\end{equation}
	with corresponding expected return $\mu_{w(\mu)}=\mu$ and standard deviation
	\[
	\sigma_{w(\mu)}=\sqrt{\frac{1}{\g}+\frac{\g}{\d}\of{\mu-\frac{\b}{\g}}^2}.
	\]
	Hence,
	\[
	\cb{(\sigma_{w(\mu)},\mu)\mid \mu\in\R}=\H_+\coloneqq \H\cap (\R_+\times\R),
	\]
	where $\H$ is a hyperbola defined by
	\[
	\H\coloneqq \cb{(\sigma,\mu)\in\R^2\mid \sigma^2-\frac{\g}{\d}\of{\mu-\frac{\b}{\g}}^2=\frac{1}{\gamma}},
	\]
	whose asymptotes are specified by the equations
	\[
	\mu=\frac{\b}{\g}\pm \sqrt{\frac{\d}{\g}}\sigma.
	\]
	In particular, for every point $(\sigma,\mu)$ on the right wing $\H_+$, there exists a unique portfolio $w\in\W$ such that $(\sigma,\mu)=(\sigma_{w},\mu_w)$.	Let $\co \H_+$ be the convex hull of $\H_+$, that is,
	\[
	\co\H_+= \cb{(\sigma,\mu)\in\R_+\times\R\mid \sigma^2-\frac{\g}{\d}\of{\mu-\frac{\b}{\g}}^2\geq \frac{1}{\gamma}}.
	\]
	For every $(\sigma,\mu)\in \co\H_+$, there exists a portfolio $w\in \W$ such that $(\sigma,\mu)=(\sigma_{w},\mu_w)$. In particular,
	\begin{equation}\label{chull}
	\cb{(\sigma_w,\mu_w)\mid w\in\W}=\co\H_+.
	\end{equation}
\end{lemma}

\begin{proof}
	These are well-known results from the analysis of the classical Markowitz problem. The reader may refer to the original derivation in \cite{merton} as well as many textbooks covering portfolio optimization, for instance, \citet[Chapter 3]{Capinski}.
	\end{proof}

Note that the point $(\frac{1}{\sqrt{\g}},\frac{\b}{\g})$ is the corner point of the right wing $\H_+$; in particular, for every $(\sigma,\mu)\in\co\H_+$, it holds $\sigma\geq \frac{1}{\sqrt{\g}}$. For each $\sigma\geq \frac{1}{\sqrt{\g}}$, let
\[
\mu(\sigma)\coloneqq \frac{\b}{\g}+\sqrt{\frac{\d}{\g}\sigma^2-\frac{\d}{\g^2}}.
\]
In particular, for every $(\sigma,\mu)\in\co\H_+$, it holds $\mu\leq \mu(\sigma)$.

\subsection{The associated two-dimensional problems}

The relation \eqref{chull} in Lemma~\ref{mw} motivates us to introduce a related problem expressed as
\begin{align*}
&\text{minimize}\; \;\;                  \rho_1\sigma-\mu\tag{$\M(r)$}\\
&\text{subject to}\;\;                     \rho_2\sigma-\mu\leq r \\
& \quad\quad\quad\quad\;\;\;      (\sigma,\mu)\in\co\H_+.
\end{align*}
Indeed, for every feasible solution $w\in\R^n$ of $(\P(r))$, the point $(\sigma_w,\mu_w)$ is a feasible solution of $(\M(r))$ and the corresponding objective function values are equal. Moreover, by the last part of Lemma~\ref{mw}, for every feasible solution $(\sigma,\mu)\in\R^2$ of $(\M(r))$, there exists a feasible solution $w\in\R^n$ of $(\P(r))$ such that $(\sigma,\mu)=(\sigma_w,\mu_w)$ and the corresponding objective function values are equal. It follows that for an optimal solution $w\in\R^n$ of $(\P(r))$, supposing that it exists, the induced feasible solution $(\sigma_w,\mu_w)$ of $(\M(r))$ is also optimal for $(\M(r))$. On the other hand, since $\rho_1\sigma-\mu\geq \rho_1\sigma-\mu(\sigma)$ and $r\geq \rho_2\sigma-\mu\geq \rho_2\sigma-\mu(\sigma)$ for every feasible solution $(\sigma,\mu)$ of $(\M(r))$, an optimal solution of $(\M(r))$, whenever it exists, must be on the upper half of $\H_+$, that is, it must be of the form $(\sigma,\mu(\sigma))$ for some $\sigma\geq \frac{1}{\sqrt{\g}}$. By the uniqueness part of Lemma~\ref{mw}, such an optimal solution corresponds to a unique portfolio given by the formula in \eqref{wfml}. Consequently, to figure out the optimal value and the possible optimal solutions of $(\P(r))$, it suffices to carry out the same analysis for $(\M(r))$ and then to recover an optimal solution of $(\P(r))$ using \eqref{wfml} whenever there is an optimal solution of $(\M(r))$.

Before providing a joint analysis of $(\M(r))$ and $(\P(r))$, we start by solving an ``unconstrained" problem, namely, the problem of minimizing the objective function of $(\M(r))$ over the whole set $\co \H_+$, without the additional risk constraint.

\begin{proposition}\label{aux}
	Consider the auxiliary problem
	\begin{align*}
	&\text{minimize}\; \;\;  \rho_1\sigma-\mu\tag{$\M_\A$}\\
	&\text{subject to}\;                      (\sigma,\mu)\in\co\H_+.
	\end{align*}
	\begin{enumerate}[(i)]
		\item Suppose that $\rho_1< \sqrt{\frac{\delta}{\g}}$. Then, $(\M_\A)$ is an unbounded problem with optimal value $-\infty$.
		\item Suppose that $\rho_1= \sqrt{\frac{\delta}{\g}}$. Then, $(\M_\A)$ has a finite infimum that is equal to $-\frac{\b}{\g}$ but the infimum is not attained by a feasible point.
		\item Suppose that $\rho_1> \sqrt{\frac{\delta}{\g}}$. Then, the unique optimal solution of $(\M_\A)$ is $(\sigma^\ast,\mu^\ast)$, where
		\begin{equation}\label{global}
		\sigma^\ast\coloneqq\frac{\rho_1}{\sqrt{\g\rho_1^2-\d}}, \quad \mu^\ast\coloneqq \frac{\b}{\g}+\frac{\d}{\g\sqrt{\g\rho_1^2-\d}}.
		\end{equation}
		Moreover, the unique portfolio $w^\ast$ with $(\sigma_{w^\ast}=\sigma^\ast,\mu_{w^\ast}=\mu^\ast)$ is given by
		\[
		w^\ast = w(\mu^\ast)=\frac{1}{\sqrt{\g\rho_1^2-\d}}C^{-1}m+\of{\frac{1}{\g}-\frac{\b}{\g\sqrt{\g\rho_1^2-\d}}}C^{-1}\1.
		\]
	\end{enumerate}
\end{proposition}

\begin{proof}
	\text{}
	\begin{enumerate}[(i)]
		\item Suppose that $\rho_1<\sqrt{\frac{\d}{\g}}$. A standard exercise in calculus yields that
		\[
		\lim_{\frac{1}{\sqrt{\g}}\leq\sigma\uparrow+\infty}\of{\rho_1\sigma-\mu(\sigma)}=\lim_{\sigma_+(r)\leq\sigma\uparrow+\infty}\of{\rho_1\sigma-\frac{\b}{\g}-\sqrt{\frac{\d}{\g}\sigma^2-\frac{\d}{\g^2}}}=-\infty.
		\]
		Since the objective function diverges to $-\infty$ on a subset of $\co\H_+$, it follows that $(\M_\A)$ is an unbounded problem with optimal value $-\infty$.
		
		\item Suppose that $\rho_1=\sqrt{\frac{\d}{\g}}$. In this case, the limit evaluated in the previous case yields
		\[
		\lim_{\frac{1}{\sqrt{\g}}\leq\sigma\uparrow+\infty}\of{\rho_1\sigma-\mu(\sigma)}=-\frac{\b}{\g}.
		\]
		On the other hand, since the hyperbola $\H$ and its asymptote $\cb{(\sigma,\mu)\in\R^2\mid \sqrt{\frac{\d}{\g}}\sigma-\mu = -\frac{\b}{\g}}$ do not intersect, there is no feasible solution $(\bar{\sigma},\bar{\mu})$ of $(\M_\A)$ such that
		\[
		\rho_1\bar{\sigma}-\bar{\mu}=\sqrt{\frac{\d}{\g}}\bar{\sigma}-\bar{\mu} =-\frac{\b}{\g}.
		\]
		Hence, the infimum of $(\M_\A)$ is equal to $-\frac{\b}{\g}$ but it is not attained by a feasible solution.
		
		\item Suppose that $\rho_1>\sqrt{\frac{\d}{\g}}$. Since every point $(
		\sigma,\mu)\in\co\H_+$ has $\rho_1\sigma-\mu\geq \rho_1\sigma-\mu(\sigma)$, if $(\sigma,\mu)$ is an optimal solution of $(\M_\A)$, then it must satisfy $\mu=\mu(\sigma)$. Moreover, since $\co\H_+$ is a convex set, by the well-known first-order condition, a point $(\sigma,\mu(\sigma))$ is an optimal solution of $(\M_\A)$ if and only if the negative of the gradient of the objective function at $(\sigma,\mu)$, which is $(-\rho_1,1)$ in this case, is a normal direction of the feasible region $\co\H_+$ at $(\sigma,\mu)$, that is,
		\[
		(-\rho_1,1)\in \cb{(x,y)\in\R^2\mid\frac{d\mu(\sigma)}{d\sigma}y+x=0, x\leq 0},
		\]
		where the derivative is calculated as
		\[
		\frac{d\mu(\sigma)}{d\sigma}=\sqrt{\frac{\d}{\g}}
		\frac{\sigma}{\sqrt{\sigma^2-\frac{1}{\g}}}.
		\]
		Hence, $(\sigma,\mu(\sigma))$ is an optimal solution of $(\M_\A)$ if and only if
		\[
		\sqrt{\frac{\d}{\g}}
		\frac{\sigma}{\sqrt{\sigma^2-\frac{1}{\g}}}=\rho_1,
		\]
		that is,
		\[
		\sigma=\sigma^\ast=\frac{\rho_1}{\sqrt{\g\rho_1^2-\g}}.
		\]
		Consequently, we also have $\mu(\sigma)=\mu^\ast$. Hence, $(\sigma^\ast,\mu^\ast)$ is the unique optimal solution of $(\M_\A)$. The corresponding portfolio $w^\ast=w(\mu^\ast)$ can be calculated easily using \eqref{wfml}.
	\end{enumerate}
\end{proof}

\subsection{Main theorems}\label{mainthm}

In this section, we present complete solutions for $(\M(r))$ and $(\P(r))$. To that end, we provide three main theorems based on the slope of the line
\[
\L(r) \coloneqq \cb{(\sigma,\mu)\in\R^2\mid \rho_2\sigma-\mu =r}
\]
related to the risk constraint. It turns out that the comparison between the slope $\rho_2$ of $\L(r)$ and the (positive) slope $\sqrt{\frac{\d}{\g}}$ of the asymptote of $\H$ is critical for the analysis.

\begin{theorem}\label{thm1}
	Let $r\in\R$ and suppose that $\rho_2<\sqrt{\frac{\delta}{\gamma}}$. Then, the hyperbola $\H$ and the line $\L(r)$ intersect at two distinct points $(\sigma_-(r),\mu_-(r))$ and $(\sigma_+(r),\mu_+(r))$ defined by
	\begin{align}
	&\sigma_{\pm}(r)\coloneqq \frac{-(\g r + \b)\rho_2\pm\sqrt{\d(\g r^2+2\b r+ \a - \rho_2^2)}}{\d - \g \rho_2^2},\label{sigmapm}\\
	&\mu_{\pm}(r)\coloneqq \frac{-\d r - \b\rho_2^2\pm\rho_2\sqrt{\d(\g r^2+2\b r+ \a - \rho_2^2)}}{\d - \g \rho_2^2}.\label{mupm}
	\end{align}
	In particular, $\sigma_+(r)>0>\sigma_-(r)$.
	Moreover, one of the following cases holds for $(\M(r))$. 
	\begin{enumerate}[(i)]
		\item Suppose that $\rho_1<\sqrt{\frac{\delta}{\gamma}}$. Then, $(\M(r))$ and $(\P(r))$ are unbounded problems with common optimal value $-\infty$.
		\item Suppose that $\rho_1=\sqrt{\frac{\delta}{\gamma}}$. Then, $(\M(r))$ and $(\P(r))$ have a common finite infimum that is equal to $-\frac{\b}{\g}$ but the infimum is not attained by a feasible solution in both problems.
		\item Suppose that $\rho_1>\sqrt{\frac{\delta}{\gamma}}$. Let 
		\begin{equation}\label{rcond}
		r^\ast \coloneqq\rho_2\sigma^\ast-\mu^\ast= \frac{\rho_1\rho_2\g-\d}{\g\sqrt{\g\rho_1^2-\d}}-\frac{\b}{\g}, \quad r_0\coloneqq \rho_2\sigma_0-
		\mu_0=\frac{\rho_2}{\sqrt{\g}}-\frac{\b}{\g}.
		\end{equation}
		It holds $r^\ast\leq r_0$. Moreover, the unique optimal solution $(\sigma^\ast,\mu^\ast)$ of $(\M_\A)$ is also the unique optimal solution of $(\M(r))$ and the corresponding portfolio $w^\ast$ is the unique optimal solution of $(\P(r))$ if and only if $r\geq r^\ast$. In particular, this is the case when $r\geq r_0$. If $r<r^\ast$, then $(\sigma_+(r),\mu_+(r))$ is the unique optimal solution of $(\M(r))$ and
		\begin{align*}
		w_+(r)\coloneqq w(\mu_+(r))&=\frac{1}{\d-\g\rho_2^2}\Bigg [\of{-\g r-\b+\frac{\g\rho_2}{\d}\sqrt{\d(\g\rho_2^2+2\b r+\a-\rho_2^2)}}C^{-1}m\\
		&\quad\quad\quad\quad\quad +\of{\b r+\a-\rho_2^2-\frac{\b\rho_2}{\d}\sqrt{\d(\g\rho_2^2+2\b r+\a-\rho_2^2)}}C^{-1}\1\Bigg ]
		\end{align*}
		is the unique optimal solution of $(\P(r))$.
		\end{enumerate}
\end{theorem}

\begin{proof}
	By the definitions of $\H$ and $\L(r)$, a point $(\sigma,\mu)\in\H\cap \L(r)$ must satisfy
	\[
	\sigma^2-\frac{\g}{\d}\of{\mu-\frac{\b}{\g}}^2=\sigma^2-\frac{\g}{\d}\of{\rho_2\sigma-r-\frac{\b}{\g}}^2=\frac{1}{\g},
	\]
	that is,
	\begin{equation}\label{quadforint}
	\of{1-\frac{\g}{\d}\rho_2^2}\sigma^2+2\frac{\g}{\d}\of{r+\frac{\b}{\g}}\rho_2\sigma-\frac{\g}{\d}\of{r+\frac{\b}{\g}}^2-\frac{1}{\g}=0.
	\end{equation}
	Note that \eqref{quadforint} is a quadratic equation in $\sigma$ whose discriminant is given by
	\begin{align}
	\Delta(r)\coloneqq  & 4\frac{\g^2}{\d^2}\of{r+\frac{\b}{\g}}^2\rho_2^2+4\of{1-\frac{\g}{\d}\rho_2^2}\sqb{ \frac{\g}{\d}\of{r+\frac{\b}{\g}}^2+\frac{1}{\g}}\notag\\
	=& 4\frac{\g}{\d}\of{r+\frac{\b}{\g}}^2+4\frac{1}{\g}-4\frac{1}{\d}\rho_2^2\label{quadratic2}\\
    =& \frac{4}{\d}\of{\g r^2 +2\b r + \frac{\b^2+\d}{\g} - \rho_2^2}\notag \\
	=& \frac{4}{\d}\of{\g r^2 +2\b r + \a - \rho_2^2}.\label{quadratic}
	\end{align}
	Using \eqref{quadratic}, one can easily check that $r\mapsto \Delta(r)$ is a strictly convex quadratic function on $\R$ whose minimum value is given by
	\begin{equation}\label{mindisc}
	\min_{r\in\R}\Delta(r)=\frac{4}{\d}\of{\frac{\d}{\g}-\rho_2^2}.
	\end{equation}
	Since $\rho_2<\sqrt{\frac{\d}{\g}}$ by assumption, we see that $\Delta(r)>0$ for every $r\in\R$ so that the quadratic equation \eqref{quadforint} has two distinct real solutions $\sigma_-(r),\sigma_+(r)$ given by \eqref{sigmapm}. Moreover, by \eqref{quadratic2} and the assumption that $\rho_2<\sqrt{\frac{\d}{\g}}$, we have
	\begin{equation}\label{rootsign}
	\d\of{\g r^2 +2\b r + \a - \rho_2^2}=\frac{\d^2}{4}\Delta(r)> \g\d\of{r+\frac{\b}{\g}}^2\geq \g^2\of{r+\frac{\b}{\g}}^2\rho_2^2=(\g r+\b)^2\rho_2^2,
	\end{equation}
	which implies that $\sigma_-(r)<0$ and $\sigma_+(r)>0$.
	The corresponding expected return values $\mu_-(r),\mu_+(r)$ given by \eqref{mupm} are calculated from the defining equation of $\L(r)$ so that
	\[
	\H\cap\L(r)=\cb{(\sigma_-(r),\mu_-(r)),(\sigma_+(r),\mu_+(r))}.
	\]
	
	Next, we consider the three possible cases for $(\M(r))$. As a preparation, we first claim that every $(\bar{\sigma},\bar{\mu})\in\H_+$ with $\bar{\sigma}\geq\sigma_+(r)$ is also a feasible solution of $(\M(r))$. In other words, we claim that the set
	\[
	\mathcal{S}\coloneqq \cb{(\bar{\sigma},\bar{\mu})\in\R_+\times\R\mid \bar{\mu}=\mu(\bar{\sigma}),\bar{\sigma}\geq \sigma_+(r)}
	\]
	is a subset of the feasible region of $(\M(r))$, that is, $\rho_2\bar{\sigma}-\bar{\mu}\leq r$ for every $(\bar{\sigma},\bar{\mu})\in\mathcal{S}$. Indeed, since $\rho_2\leq \sqrt{\frac{\d}{\g}}$, we have
	\begin{equation}\label{deriv2}
	\frac{d}{d\sigma}\of{\rho_2\sigma-\mu(\sigma)}
    =\rho_2-\frac{\sqrt{\frac{\d}{\g}}\sigma}{\sqrt{\sigma^2-\frac{1}{\g}}}
	\leq \sqrt{\frac{\d}{\g}}\of{1-\frac{\sigma}{\sqrt{\sigma^2-\frac{1}{\g}}}}< 0
	\end{equation}
	for every $\sigma>\frac{1}{\sqrt{\g}}$. Since we also have $\mu(\sigma_+(r))\geq \mu_+(r)$, it follows that every $(\bar{\sigma},\bar{\mu})\in\mathcal{S}$ satisfies
	\[
	r= \rho_2\sigma_+(r)-\mu_+(r)\geq  \rho_2\sigma_+(r)-\mu(\sigma_+(r))>\rho_2\bar{\sigma}-\mu(\bar{\sigma})=\rho_2\bar{\sigma}-\bar{\mu}
	\]
	so that it is feasible for $(\M(r))$. Hence, the claim follows.
	\begin{enumerate}[(i)]
		\item Suppose that $\rho_1<\sqrt{\frac{\d}{\g}}$. A standard exercise in calculus yields that
		\[
		\lim_{\sigma_+(r)\leq\sigma\uparrow+\infty}\of{\rho_1\sigma-\mu(\sigma)}=\lim_{\sigma_+(r)\leq\sigma\uparrow+\infty}\of{\rho_1\sigma-\frac{\b}{\g}-\sqrt{\frac{\d}{\g}\sigma^2-\frac{\d}{\g^2}}}=-\infty.
		\]
		Since the objective function diverges to $-\infty$ on $\mathcal{S}$, it follows that $(\M(r))$ is an unbounded problem with optimal value $-\infty$.
		
		\item Suppose that $\rho_1=\sqrt{\frac{\d}{\g}}$. In this case, the limit evaluated in the previous case yields
		\[
		\lim_{\sigma_+(r)\leq\sigma\uparrow+\infty}\of{\rho_1\sigma-\mu(\sigma)}=-\frac{\b}{\g}.
		\]
		On the other hand, since the hyperbola $\H$ and its asymptote $\cb{(\sigma,\mu)\in\R^2\mid \sqrt{\frac{\d}{\g}}\sigma-\mu = -\frac{\b}{\g}}$ do not intersect, there is no feasible solution $(\sigma,\mu)$ of $(\M(r))$ such that
		\[
		\rho_1\sigma-\mu=\sqrt{\frac{\d}{\g}}\sigma-\mu =-\frac{\b}{\g}.
		\]
		Hence, the infimum in $(\M(r))$ is equal to $-\frac{\b}{\g}$ but it is not attained by a feasible solution.
		
		\item Suppose that $\rho_1>\sqrt{\frac{\d}{\g}}$. Note that the feasible region of $(\M(r))$ is a subset of that of $(\M_\A)$. Hence, in view of Proposition~\ref{aux}, the unique optimal solution $(\sigma^\ast,\mu^\ast)$ of $(\M_\A)$ is also the unique optimal solution of $(\M(r))$ if and only if it is feasible for $(\M(r))$, that is,
		\[
		r^\ast=\rho_2\sigma^\ast-\mu^\ast \leq r,
		\]
		where $r^\ast$ is defined by \eqref{rcond}.
		
		Next, we show that $r^\ast\leq r_0$, where $r_0$ is defined by \eqref{rcond}. So we show that
		\[
		\frac{\rho_1\rho_2\g-\d}{\g\sqrt{\g\rho_1^2-\d}}-\frac{\b}{\g}\leq \frac{\rho_2}{\sqrt{\g}}-\frac{\b}{\g},
		\]
		which is equivalent to
		\begin{equation}\label{somesuff}
		\rho_1\rho_2\g-\d\leq \rho_2\sqrt{\g}\sqrt{\g\rho_1^2-\d}.
		\end{equation}
		If $\g\rho_1\rho_2-\d\leq 0$, then \eqref{somesuff} holds trivially. Suppose that $\g\rho_1\rho_2-\d> 0$. In this case, \eqref{somesuff} is equivalent to
		\[
		\g^2\rho_1^2\rho_2^2+\d^2-2\g\d\rho_1\rho_2=(\g\rho_1\rho_2-\d)^2\leq \g\rho_2^2\of{\g\rho_1^2-\d}=\g^2\rho_1^2\rho_2^2-\g\d\rho_2^2,
		\]
		which is equivalent to
		\[
		\d-2\g\rho_1\rho_2+\g\rho_2^2\leq 0.
		\]
		But the last inequality follows from the supposition and the assumption that $\rho_1>\sqrt{\frac{\d}{\g}}>\rho_2$ since
		\[
		\d-2\g\rho_1\rho_2+\g\rho_2^2\leq \g\rho_1\rho_2-2\g\rho_1\rho_2+\g\rho_2^2=\g\rho_2(\rho_2-\rho_1)\leq 0.
		\]
		Consequently, \eqref{somesuff} holds when $\g\rho_1\rho_2-\d>0$ as well. Hence, $r^\ast\leq r_0$.
		
		Finally, we consider the case $r<r^\ast$, that is, $(\sigma^\ast,\mu^\ast)$ is not feasible for $(\M(r))$. In this case, we prove that $(\sigma_+(r),\mu_+(r))$ is the unique optimal solution of $(\M(r))$. To that end, note that we have $r<r_0$ in this case so that
		\[
		\rho_2\sigma_+(r)-\mu_+(r)=r<r_0=\rho_2\sigma_0-\mu_0\leq \rho_2\sigma_+(r)-\mu_0.
		\]
		This implies $\mu_+(r)>\mu_0$. In particular, $\mu_+(r)=\mu(\sigma_+(r))$. Next, let $(\bar{\sigma},\bar{\mu})$ be a feasible solution of $(\M(r))$ with $(\bar{\sigma},\bar{\mu})\neq(\sigma_+(r),\mu_+(r))$. We first claim that $\bar{\sigma}>\sigma_+(r)$. To get a contradiction, suppose $\bar{\sigma}\leq \sigma_+(r)$. By \eqref{deriv2}, $\sigma\mapsto \rho_2\sigma-\mu(\sigma)$ is a decreasing function. Using this and the fact that $\bar{\mu}\leq\mu(\bar{\sigma})$, we obtain
	    \[
	    r\geq \rho_2\bar{\sigma}-\bar{\mu}\geq \rho_2\bar{\sigma}-\mu(\bar{\sigma})\geq \rho_2\sigma_+(r)-\mu_+(r)=r,
	    \]
	    which yields $\rho_2\bar{\sigma}-\bar{\mu}=r$ and $\bar{\mu}=\mu(\bar{\sigma})$. This implies $(\bar{\sigma},\bar{\mu})\in\H\cap\L(r)$ and hence $(\bar{\sigma},\bar{\mu})=(\sigma_+(r),\mu_+(r))$, which is a contradiction. Hence, the claim follows. On the other hand, using the assumption $\rho_1>\sqrt{\frac{\d}{\g}}$, we notice that
	    \begin{equation}\label{strictlyinc}
	    \frac{d}{d\sigma}\of{\rho_1\sigma-\mu(\sigma)}
	    =\rho_1-\frac{\sqrt{\frac{\d}{\g}}\sigma}{\sqrt{\sigma^2-\frac{1}{\g}}}>0 \quad \Leftrightarrow\quad \sigma>\frac{\rho_1}{\sqrt{\g\rho_1^2-\d}}=\sigma^\ast,
	    \end{equation}
	    that is $\sigma\mapsto \rho_1\sigma-\mu(\sigma)$ is a strictly increasing function for $\sigma>\sigma^\ast$. Moreover, we have $\bar{\sigma}>\sigma_+(r)>\sigma^\ast$. Indeed, the first inequality is by the previous claim. The second inequality holds as otherwise, $(\sigma^\ast,\mu^\ast)$ would be feasible for $(\M(r))$ by the preparatory claim preceding the analysis of the three cases, which is excluded by the assumption $r<r^\ast$. Since we also have $\bar{\mu}\leq\mu(\bar{\sigma})$ and $\mu_+(r)=\mu(\sigma_+(r))$, it follows that
	    \[
	    \rho_1\bar{\sigma}-\bar{\mu}\geq \rho_1\bar{\sigma}-\mu(\bar{\sigma})>\rho_1\sigma_+(r)-\mu(\sigma_+(r))=\rho_1\sigma_+(r)-\mu_+(r),
	    \]
	    that is, $(\bar{\sigma},\bar{\mu})$ is not optimal for $(\M(r))$. Hence, $(\sigma_+(r),\mu_+(r))$ is the unique optimal solution of $(\M(r))$.
	\end{enumerate}
	\end{proof}

\begin{theorem}\label{thm2}
	Let $r\in\R$ and suppose that $\rho_2>\sqrt{\frac{\delta}{\gamma}}$. Then, the hyperbola $\H$ and the line $\L(r)$ intersect precisely at two points, $(\sigma_-(r),\mu_-(r))$ and $(\sigma_+(r),\mu_+(r))$ defined by \eqref{sigmapm}, if and only if $r\leq r_{-}$ or $r\geq r_+$, where
	\begin{equation}\label{rpm}
	r_{\pm}\coloneqq \frac{-\b\pm\sqrt{\g\rho_2^2-\d}}{\g}.
	\end{equation}
	In particular, it holds $\sigma_+(r)\leq\sigma_-(r)<0$ if $r\leq r_-$, it holds $0<\sigma_+(r)\leq\sigma_-(r)$ if $r\geq r_+$. Moreover, the points $(\sigma_-(r),\mu_-(r))$ and $(\sigma_+(r),\mu_+(r))$ are identical if and only if $r=r_-$ or $r=r_+$. The hyperbola $\H$ and the line $\L(r)$ do not intersect at all if and only if $r_-<r<r_+$. Consequently, $(\M(r))$ is feasible if and only if $r\geq r_+$.
	
	Suppose that $r\geq r_+$. Then, one of the following cases holds for $(\M(r))$.
	\begin{enumerate}[(i)]
		\item Suppose that $\rho_1\leq \sqrt{\frac{\delta}{\gamma}}$. Then, $(\sigma_-(r),\mu_-(r))$ is the unique optimal solution of $(\M(r))$ and
		\begin{align*}
		w_-(r)\coloneqq w(\mu_-(r))&=\frac{1}{\g\rho_2^2-\d}\Bigg [\of{\g r+\b+\frac{\g\rho_2}{\d}\sqrt{\d(\g\rho_2^2+2\b r+\a-\rho_2^2)}}C^{-1}m\\
		&\quad\quad\quad\quad\quad +\of{-\b r-\a+\rho_2^2-\frac{\b\rho_2}{\d}\sqrt{\d(\g\rho_2^2+2\b r+\a-\rho_2^2)}}C^{-1}\1\Bigg ]
		\end{align*}
		is the unique optimal solution of $(\P(r))$.
		\item Suppose that $\rho_1>\sqrt{\frac{\delta}{\gamma}}$. Then, the unique optimal solution $(\sigma^\ast,\mu^\ast)$ of $(\M_\A)$ is also the unique optimal solution of $(\M(r))$ and the corresponding portfolio $w^\ast$ is the unique optimal solution of $(\P(r))$ if and only if $r\geq r^\ast$, where $r^\ast$ is defined by \eqref{rcond}.
		
		It holds $r_+=r^\ast$ if $\rho_1=\rho_2$ and $r_+<r^\ast$ if $\rho_1\neq\rho_2$. Suppose that $\rho_1\neq \rho_2$ and $r_+\leq r<r^\ast$. Then, one of the following cases holds for $(\M(r))$.
		\begin{enumerate}[a.]
			\item If $\rho_1<\rho_2$, then $(\sigma_-(r),\mu_-(r))$ is the unique optimal solution of $(\M(r))$ and $w_{-}(r)$ is the unique optimal solution of $(\P(r))$.
			\item If $\rho_1>\rho_2$, then $(\sigma_+(r),\mu_+(r))$ is the unique optimal solution of $(\M(r))$ and $w_{+}(r)$ is the unique optimal solution of $(\P(r))$.
		\end{enumerate}
	\end{enumerate}
\end{theorem}

\begin{proof}
	As in the proof of Theorem~\ref{thm1}, a point $(\sigma,\mu)\in\H\cap \L(r)$ satisfies \eqref{quadforint}, which is a quadratic equation in $\sigma$ with discriminant $\Delta(r)$ given by \eqref{quadratic}. However, since $\rho_2>\sqrt{\frac{\delta}{\gamma}}$, the minimum in \eqref{mindisc} is strictly negative: $\min_{r\in\R}\Delta(r)<0$. Moreover, we have $\Delta(r)=0$ if and only if $r\in\cb{r_-,r_+}$, where $r_{\pm}$ are defined by \eqref{rpm}; $\Delta(r)<0$ if and only if $r_-<r<r_+$; $\Delta(r)>0$ if and only if $r<r_-$ or $r>r_+$. Hence, $\H\cap\L(r)$ is nonempty if and only if $r\leq r_-$ or $r\geq r_+$, and the intersection consists of $(\sigma_-(r),\mu_-(r)), (\sigma_+(r),\mu_+(r))$ in this case. Mimicing the arguments for \eqref{rootsign}, this time with $\rho_2>\sqrt{\frac{\d}{\g}}$, gives
	\[
	\d(\g r^2+2\b r+\a - \rho_2^2)<(\g r+\b)^2 \rho_2^2.
	\]
	It follows that $\sigma_+(r)\leq\sigma_-(r)<0$ if $r\leq r_-\leq -\frac{\b}{\g}$ and $0<\sigma_+(r)\leq\sigma_-(r)$ if $r\geq r_+\geq -\frac{\b}{\g}$.
	The rest of the claims in the first paragraph of the theorem follows immediately.
	
	For the rest of the proof, suppose that $r\geq r_+$. We consider the three possible cases for $(\M(r))$ next. As a preparation, we first show that $\mu(\sigma_-(r))=\mu_-(r)$. To that end, it suffices to show that
	\[
	\frac{\d r + \b\rho_2^2+\rho_2\sqrt{\d(\g r^2+2\b r+ \a - \rho_2^2)}}{ \g \rho_2^2-\d}=\mu_-(r)\geq \mu_0=\frac{\b}{\g},
	\]
	which is equivalent to
	\[
	\d\g r+\g\b\rho_2^2+\g\rho_2\sqrt{\d(\g r^2+2\b r+\a-\rho_2^2)}\geq \g\b\rho_2^2-\d\b,
	\]
	that is,
	\begin{equation}\label{sthtoshow}
    \d(\g r+\b)+\g\rho_2\sqrt{\d(\g r^2+2\b r+\a-\rho_2^2)}\geq 0.
	\end{equation}
	On the other hand, since $r\geq r_+$, we have $\g r+\b\geq \sqrt{\g\rho_2^2-\d}>0$ from which \eqref{sthtoshow} follows. Hence, $\mu(\sigma_-(r))=\mu_-(r)$.
	
	\begin{enumerate}[(i)]
		\item Suppose that $\rho_1\leq\sqrt{\frac{\d}{\g}}$. Let $(\bar{\sigma},\bar{\mu})$ be a feasible solution of $(\M(r))$ with $(\bar{\sigma},\bar{\mu})\neq (\sigma_-(r),\mu_-(r))$. We claim that $\bar{\sigma}\leq \sigma_-(r)$. To get a contradiction, suppose that $\bar{\sigma}>\sigma_-(r)$. Similar to \eqref{strictlyinc}, we notice that $\sigma\mapsto\rho_2\sigma-\mu(\sigma)$ is a strictly increasing function for $\sigma>\frac{\rho_2}{\sqrt{\g\rho_2^2-\d}}$ thanks to the assumption $\rho_2>\sqrt{\frac{\d}{\g}}$. On the other hand, since we assume that $r\geq r_+$, we have
		\[
		\sigma_-(r)=\frac{(\g r+\b)\rho_2+\sqrt{\d(\g r^2+2\b r+\a-\rho_2^2)}}{\g\rho_2^2-\d}\geq \frac{(\g r_+ +\b)\rho_2}{\g\rho_2^2-\d}=\frac{\rho_2}{\sqrt{\g\rho_2^2-\d}}.
		\]
		Since we also have $\bar{\mu}\leq \mu(\bar{\sigma})$ and $\mu(\sigma_-(r))=\mu_-(r)$, it follows that
		\[
		r\geq \rho_2\bar{\sigma}-\bar{\mu}\geq \rho_2\bar{\sigma}-\mu(\bar{\sigma})> \rho_2\sigma_-(r)-\mu(\sigma_-(r))=\rho_2\sigma_-(r)-\mu_-(r)=r,
		\]
		which is a contradiction. Hence, $\bar{\sigma}\leq \sigma_-(r)$. Moreover, we further have $\bar{\sigma}<\sigma_-(r)$ as otherwise $\bar{\sigma}=\sigma_-(r)$ would imply
		\[
		r\geq \rho_2\bar{\sigma}-\bar{\mu}=\rho_2\sigma_-(r)-\bar{\mu}\geq\rho_2\sigma_-(r)-\mu(\bar{\sigma})=\rho_2\sigma_-(r)-\mu_-(r)=r
		\]
		so that $(\bar{\sigma},\bar{\mu})=(\sigma_-(r),\mu_-(r))$, which is a contradiction. On the other hand, similar to \eqref{deriv2}, we can argue that $\sigma\mapsto\rho_1\sigma-\mu(\sigma)$ is a strictly decreasing function for $\sigma>\frac{1}{\sqrt{\g}}$ thanks to the assumption $\rho_1\leq\sqrt{\frac{\d}{\g}}$. Consequently, $\bar{\sigma}<\sigma_-(r)$ implies
		\[
		\rho_1\bar{\sigma}-\bar{\mu}\geq \rho_1\bar{\sigma}-\mu(\bar{\sigma})>\rho_1\sigma_-(r)-\mu(\sigma_-(r))=\rho_1\sigma_-(r)-\mu_-(r),
		\]
		that is, $(\bar{\sigma},\bar{\mu})$ is not optimal for $(\M(r))$. Hence, $(\sigma_-(r),\mu_-(r))$ is the unique optimal solution of $(\M(r))$.
		
		\item Suppose that $\rho_1>\sqrt{\frac{\d}{\g}}$. In this case, as in the proof of (iii) of Theorem~\ref{thm3}, we note that the unique optimal solution $(\sigma^\ast,
		\mu^\ast)$ of $(\M_\A)$ is also the unique optimal solution of $(\M(r))$ if and only if $r\geq r^\ast$, where $r^\ast$ is defined by \eqref{rcond}. From the definitions, it is clear that $r_+=r^\ast$ if $\rho_1=\rho_2$. Suppose that $\rho_1\neq\rho_2$. We first claim that
		\begin{equation}\label{critical}
		r_+<r^\ast.
			\end{equation}
			Indeed, supposing otherwise would yield $\sqrt{(\g\rho_1^2-\d)(\g\rho_2^2-\d)}\geq \rho_1\rho_2\g-\d$, which is equivalent to $(\g\rho_1^2-\d)(\g\rho_2^2-\d)\geq (\rho_1\rho_2\g-\d)^2$ as we have $\rho_1\rho_2>\frac{\d}{\g}$ by the assumptions on $\rho_1, \rho_2$. Further simplification would yield the inequality $0\geq \g\d(\rho_1-\rho_2)^2$, which is a contradiction since $\rho_1\neq\rho_2$ by supposition. Hence, the claim follows.
			
			In view of \eqref{critical}, it remains to figure out the optimal solution of $(\M(r))$ under the condition that
			\begin{equation}\label{rem}
			r_+\leq r<r^\ast,
			\end{equation}
			which we assume for the rest of the proof.
			
			\begin{enumerate}[a.]
			\item Let us assume that $\rho_1<\rho_2$. We prove that $(\sigma_-(r),\mu_-(r))$ is the unique optimal solution of $(\M(r))$. To that end, let $(\bar{\sigma},\bar{\mu})$ be a feasible solution of $(\M(r))$ such that $(\bar{\sigma},\bar{\mu})\neq (\sigma_-(r),\mu_-(r))$. Following the same arguments as in the proof of (i), one can check that $\bar{\sigma}<\sigma_-(r)$.
			
			Next, we show $\sigma_-(r)<\sigma^\ast$. Note that $\sigma_-(r)<\sigma^\ast$ is equivalent to
			\[
			\frac{(\g r+\b)\rho_2+\sqrt{\d(\g r^2+2\b r+\a-\rho_2^2)}}{\g\rho_2^2-\d}<\frac{\rho_1}{\sqrt{\g\rho_1^2-\d}},
			\]
			that is,
			\begin{equation}\label{toshow1}
			\sqrt{\d(\g r^2+2\b r+\a-\rho_2^2)}<	(\g r+\b)\rho_2-\frac{\rho_1(\g\rho_2^2-\d)}{\sqrt{\g\rho_1^2-\d}}.
			\end{equation}
			However, \eqref{toshow1} does not hold true when its right hand side is negative, that is, when
			\begin{equation*}
			r\geq r_{++}\coloneqq \frac{\rho_1(\g\rho_2^2-\d)}{\g\rho_2\sqrt{\g\rho_1^2-\d}}-\frac{\b}{\g}.
			\end{equation*}
			On the other hand, note that
			\begin{align}\label{rcomp}
			r^\ast<r_{++}&\quad\Leftrightarrow\quad\frac{\rho_1(\g\rho_2^2-\d)}{\g\rho_2\sqrt{\g\rho_1^2-\d}}-\frac{\b}{\g}>\frac{\rho_1\rho_2\g-\d}{\g\sqrt{\g\rho_1^2-\d}}-\frac{\b}{\g}\notag\\
			&\quad\Leftrightarrow\quad \g\rho_1\rho_2^2-\rho_1\d>\g\rho_1\rho_2^2-\d\rho_2\notag\\
			&\quad\Leftrightarrow\quad \rho_1<\rho_2.
			\end{align}
			Hence, in view of \eqref{rem}, we always have $r<r_{++}$ so that \eqref{toshow1} can be rewritten as
			\begin{equation}\label{sqrd}
			\d(\g r^2+2\b r+\a-\rho_2^2)<(\g r+\b)^2\rho_2^2+\frac{\rho_1^2(\g\rho_2^2-\d)^2}{(\g\rho_1^2-\d)}-2(\g r+\b)\frac{\rho_1\rho_2(\g\rho_2^2-\d)}{\sqrt{\g\rho_1^2-\d}},
			\end{equation}
			which is equivalent to
			\begin{equation}\label{quadr1}
			0<\g r^2+2\of{\b-\frac{\g\rho_1\rho_2}{\sqrt{\g\rho_1^2-\d}}}r+\a+\frac{\rho_1^2(\g\rho_2^2-\d)}{\g\rho_1^2-\d}-\frac{2\b\rho_1\rho_2}{\sqrt{\g\rho_1^2-\d}}.
			\end{equation}
			The discriminant of the quadratic function of $r$ in the right hand side of \eqref{quadr1} is calculated as $\frac{4\d^2}{\g\rho_1^2-\d}>0$ so that this function has two distict real zeros and the smaller of these zeros is precisely $r^\ast$. Since we assume \eqref{rem}, \eqref{quadr1} always holds and we have $\sigma_-(r)<\sigma^\ast$.
			
			On the other hand, by \eqref{strictlyinc}, $\sigma\mapsto \rho_1\sigma-\mu(\sigma)$ is a strictly decreasing function for $\sigma<\sigma^\ast$. Hence, $\bar{\sigma}<\sigma_-(r)<\sigma^\ast$ implies
			\begin{equation}\label{proofofopt}
			\rho_1\bar{\sigma}-\bar{\mu}\geq	\rho_1\bar{\sigma}-\mu(\bar{\sigma})> \rho_1\sigma_-(r)-\mu(\sigma_-(r))=\rho_1\sigma_-(r)-\mu_-(r).
			\end{equation}
			We conclude that $(\bar{\sigma},\bar{\mu})$ is not optimal for $(\M(r))$. Hence, $(\sigma_-(r),\mu_-(r))$ is the unique optimal solution.
			
			\item Let us assume that $\rho_1>\rho_2$. We prove that $(\sigma_+(r),\mu_+(r))$ is the unique optimal solution of $(\M(r))$. To that end, let $(\bar{\sigma},\bar{\mu})$ be a feasible solution of $(\M(r))$ with $(\bar{\sigma},\bar{\mu})\neq (\sigma_+(r),\mu_+(r))$. We claim that $\bar{\sigma}\geq \sigma_+(r)$. To get a contradiction, suppose that $\bar{\sigma}<\sigma_+(r)$. 
			
			Similar to the proof of (i), we notice that $\sigma\mapsto\rho_2\sigma-\mu(\sigma)$ is a strictly decreasing function for $\sigma<\frac{\rho_2}{\sqrt{\g\rho_2^2-\d}}$. Next, we show that
			\begin{equation}\label{plusthreshold}
			\frac{(\g r+\b)\rho_2-\sqrt{\d(\g r^2+2\b r+\a-\rho_2^2)}}{\g\rho_2^2-\d}=\sigma_+(r)\leq\frac{\rho_2}{\sqrt{\g\rho_2^2-\d}},
			\end{equation}
			which is equivalent to
			\begin{equation}\label{befsq}
			(\g r+\b)\rho_2-\rho_2\sqrt{\g\rho_2^2-\d}\leq\sqrt{\d(\g r^2+2\b r+\a-\rho_2^2)}
			\end{equation}
			It is easy to check that the left hand side of \eqref{befsq} is positive thanks to the assumption $r\geq r_+$. Hence, \eqref{befsq} is equivalent to
			\[
			\rho_2^2\of{\g^2 r^2+2\b r+\b^2}+\rho_2^2(\g\rho_2^2-\d)-2\rho_2^2\sqrt{\g\rho_2^2-\d}(\g r+\b)\leq \d (\g r^2+2\b r+\a-\rho_2^2),
			\]
			that is,
			\begin{equation}\label{quad2}
			(\g\rho_2^2-\d)\g r^2+2\of{\b(\g\rho_2^2-\d)-\g\rho_2^2\sqrt{\g\rho_2^2-\d}}r+\rho_2^2(\b^2+\g\rho_2^2)-\a\d-2\b\rho_2^2\sqrt{\g\rho_2^2-\d}\leq 0.
			\end{equation}
			One can check that the discriminant of the quadratic function of $r$ on the left hand side of \eqref{quad2} is calculated as $4\d^2(\g\rho_2^2-\d)>0$ so that it has two distinct zeros which are given as
			\[
			-\frac{\b}{\g}+\frac{\g\rho_2^2\pm\d}{\g\sqrt{\g\rho_2^2-\d}}.
			\]
			Hence, \eqref{quad2} holds if and only if $r$ is between these two zeros. Note that the smaller zero is equal to $r_+$ and we have $r\geq r_+$ by assumption. Next, we show that
			\begin{equation}\label{bigzero}
			r^\ast=-\frac{\b}{\g}+\frac{\rho_1\rho_2\g-\d}{\g\sqrt{\g\rho_1^2-\d}}<-\frac{\b}{\g}+\frac{\g\rho_2^2+\d}{\g\sqrt{\g\rho_2^2-\d}},
			\end{equation}
			which is equivalent to
			\begin{equation}\label{bigzero2}
			\sqrt{\g\rho_1^2-\d}(\g\rho_2^2+\d)>(\rho_1\rho_2\g-\d)\sqrt{\g\rho_2^2-\d}.
			\end{equation}
			By the assumptions $\rho_1>\sqrt{\frac{\d}{\g}}$, $\rho_2>\sqrt{\frac{\d}{\g}}$, the right hand side of \eqref{bigzero2} is positive so that \eqref{bigzero2} is equivalent to
			\[
			(\g\rho_1^2-\d)(\d^2\rho_2^4+\d^2+2\g\d\rho_2^2)>(\g\rho_2^2-\d)(\g^2\rho_1^2\rho_2^2+\d^2-2\g\d\rho_1\rho_2),
			\]
			that is,
			\[
			\d\rho_1^2+3\g\rho_1^2\rho_2^2-\g\rho_2^4-3\d\rho_2^2+2\g\rho_1\rho_2^3-2\d\rho_1\rho_2=\d\rho_1^2-\g\rho_2^4+3\rho_2^2(\g\rho_1^2-\d)+2\rho_1\rho_2(\g\rho_2^2-\d)>0.
			\]
			However, since $\rho_1>\rho_2$, we have
			\begin{align*}
			\d\rho_1^2-\g\rho_2^4+3\rho_2^2(\g\rho_1^2-\d)+2\rho_1\rho_2(\g\rho_2^2-\d)&>\d\rho_2^2-\g\rho_2^4+3\rho_2^2(\g\rho_2^2-\d)+2\rho^2_2(\g\rho_2^2-\d)\\
			&=\rho_2^2(\d-\g\rho_2^2)+5\rho_2^2(\g\rho_2^2-\d)\\
			&=4\rho_2^2(\g\rho_2^2-\d)>0
			\end{align*}
			so that \eqref{bigzero} holds. Consequently, the assumption \eqref{rem} guarantees that \eqref{plusthreshold} holds.
			
			Next, we show that $\mu(\sigma_+(r))=\mu_+(r)$. To that end, it suffices to show that
			\begin{equation}\label{muplus}
			\frac{\d r + \b\rho_2^2-\rho_2\sqrt{\d(\g r^2+2\b r+ \a - \rho_2^2)}}{ \g \rho_2^2-\d}=\mu_+(r)\geq \mu_0=\frac{\b}{\g},
			\end{equation}
			which is equivalent to
			\[
			\d\g r+\g\b\rho_2^2-\g\rho_2\sqrt{\d(\g r^2+2\b r+\a-\rho_2^2)}\geq \g\b\rho_2^2-\d\b,
			\]
			that is,
			\begin{equation}\label{sthtoshow1}
			\sqrt{\d}(\g r+\b)\geq\g\rho_2\sqrt{\g r^2+2\b r+\a-\rho_2^2}.
			\end{equation}
			On the other hand, since $r\geq r_+$, we have $\g r+\b\geq \sqrt{\g\rho_2^2-\d}>0$ so that \eqref{sthtoshow1} is equivalent to
			\begin{align}
			0&\geq -\d(\g^2 r^2+\b^2+2\g\b r)+\g^2\rho_2^2(\g r^2+2\b r+\a-\rho_2^2)\label{quad31}\\
			&=(\g\rho_2^2-\d)\g(\g r^2+2\b r)-\d\b^2+\g^2\rho_2^2(\a-\rho_2^2).\label{quad3}
			\end{align}
			Note that the discriminant of the quadratic function of $r$ in \eqref{quad3} is $4\g^3\rho_2^2(\g\rho_2^2-\d)^2>0$ so that it has two distinct zeros given by
			\[
			-\frac{\b}{\g}\pm\frac{\rho_2}{\sqrt{\g}}.
			\]
			Hence, the inequality in \eqref{quad31} holds if and only if $r$ is between these two zeros. Clearly, we have $-\frac{\b}{\g}-\frac{\rho_2}{\sqrt{\g}}<-\frac{\b}{\g}\leq r_+\leq r$. On the other hand, note that the larger zero is equal to $r_0$ and it is easy to check that $r^\ast\leq r_0$ if and only if $\rho_1\geq \frac{\g\rho_2^2+\d}{2\g\rho_2}$ and we also have $\rho_2>\frac{\g\rho_2^2+\d}{2\g\rho_2}$ since $\rho_2>\sqrt{\d}{\g}$. Consequently, the assumption $\rho_1>\rho_2$ implies that $r^\ast\leq r_0$ holds so that $r$ is between the two zeros of the quadratic function in \eqref{quad3}. Hence, the inequality in \eqref{quad31} holds, \eqref{muplus} holds and we have $\mu(\sigma_+(r))=\mu_+(r)$.
			
			Hence, $\bar{\sigma}<\sigma_+(r)\leq\frac{\rho_2}{\sqrt{\g\rho_2^2-\d}}$ and $\mu(\sigma_+(r))=\mu_+(r)$ imply that
			\[
			r\geq \rho_2\bar{\sigma}-\bar{\mu}\geq \rho_2\bar{\sigma}-\mu(\bar{\sigma})> \rho_2\sigma_+(r)-\mu(\sigma_+(r))=\rho_2\sigma_+(r)-\mu_+(r)=r,
			\]
			which is a contradiction. Hence, $\bar{\sigma}\geq \sigma_+(r)$. Moreover, we further have $\bar{\sigma}>\sigma_+(r)$ as otherwise $\bar{\sigma}=\sigma_+(r)$ would imply
			\[
			r\geq \rho_2\bar{\sigma}-\bar{\mu}=\rho_2\sigma_+(r)-\bar{\mu}\geq\rho_2\sigma_+(r)-\mu(\bar{\sigma})=\rho_2\sigma_+(r)-\mu_+(r)=r
			\]
			so that $(\bar{\sigma},\bar{\mu})=(\sigma_+(r),\mu_+(r))$, which is a contradiction.
			
			In view of \eqref{rcomp}, we have $r^\ast>r_{++}$. We show that $\sigma_+(r)>\sigma^\ast$. Note that $\sigma_+(r)>\sigma^\ast$ is equivalent to
			\[
			\frac{(\g r+\b)\rho_2-\sqrt{\d(\g r^2+2\b r+\a-\rho_2^2)}}{\g\rho_2^2-\d}>\frac{\rho_1}{\sqrt{\g\rho_1^2-\d}},
			\]
			that is,
			\begin{equation}\label{toshow}
			\sqrt{\d(\g r^2+2\b r+\a-\rho_2^2)}<	(\g r+\b)\rho_2-\frac{\rho_1(\g\rho_2^2-\d)}{\sqrt{\g\rho_1^2-\d}}.
			\end{equation}
			Since $r^\ast>r_{++}$, the right hand side of \eqref{toshow} is strictly positive so that \eqref{toshow} is equivalent to \eqref{sqrd} as well as to \eqref{quadr1}. Repeating the same analysis of the quadratic function in \eqref{quadr1}, we see that this function has two distinct real zeros and the smaller of these zeros is precisely $r^\ast$. Since we assume \eqref{rem}, \eqref{quadr1} always holds and we have $\sigma_+(r)>\sigma^\ast$.
			
			Consequently, $\bar{\sigma}>\sigma_+(r)>\sigma^\ast$ and the fact that $\sigma\mapsto\rho_1\sigma-\mu(\sigma)$ is strictly increasing for $\sigma>\sigma^\ast$ imply
			\[
			\rho_1\bar{\sigma}-\bar{\mu}\geq \rho_1\bar{\sigma}-\mu(\bar{\sigma})>\rho_1\sigma_+(r)-\mu(\sigma_+(r))=\rho_1\sigma_+(r)-\mu_+(r),
			\]
			that is, $(\bar{\sigma},\bar{\mu})$ is not optimal for $(\M(r))$. Hence, $(\sigma_+(r),\mu_+(r))$ is the unique optimal solution of $(\M(r))$.
			\end{enumerate}
	\end{enumerate}
	\end{proof}

\begin{remark}
	Theorem~\ref{thm2} can be used to understand if solving a risk-risk problem is significantly different from solving a mean-risk problem in the following way. Suppose that the risk measure in the constraint satisfies $\rho_2>\sqrt{\frac{\d}{\g}}$ and we have $r\geq r_+$. When the risk measure in the objective is simply the negative expected value, we have $\rho_1=0$. In this case, by (i) of Theorem~\ref{thm2}, $w_-(r)$ is the unique optimal portfolio. However, the optimal portfolio does not change at all even if we use a nontrivial coherent risk measure in the objective such as value-at-risk or average value-at-risk as long as $\rho_1\leq \sqrt{\frac{\d}{\g}}$. Hence, the risk-risk problem with $\rho_1\leq \sqrt{\frac{\d}{\g}}$ is practically the same as the mean-risk problem. On the other hand, by (ii) of Theorem~\ref{thm2}, for $\rho_1>\sqrt{\frac{\d}{\g}}$ with $r\geq r^\ast$, the risk-risk problem becomes significantly different from the mean-risk problem since the unique optimal portfolio is $w^\ast$, which depends on the choices of $\rho_1$ and $\rho_2$.
	\end{remark}

\begin{theorem}\label{thm3}
	Let $r\in\R$ and suppose that $\rho_2=\sqrt{\frac{\delta}{\gamma}}$. Then, the hyperbola $\H$ and the line $\L(r)$ intersect precisely at the single point $(\hat{\sigma}(r),\hat{\mu}(r))$ defined by
	\begin{equation}\label{linsigma}
	\hat{\sigma}(r)\coloneqq\frac{\g r^2+2\b r+\a}{2\rho_2(\g r+\b)},\quad \hat{\mu}(r)\coloneqq\frac{\a-\g r^2}{2(\g r+\b)}
	\end{equation}
	if and only if $r\neq -\frac{\b}{\g}$. In particular, $\bar{\sigma}(r)<0$ if $r<-\frac{\b}{\g}$, it holds $\bar{\sigma}(r)>0$ if $r>-\frac{\b}{\g}$. The hyperbola $\H$ and the line $\L(r)$ do not intersect at all if and only if $r=-\frac{\b}{\g}$. Consequently, $(\M(r))$ is feasible if and only if $r>-\frac{\b}{\g}$. 
	
	Suppose that $r>-\frac{\b}{\g}$. Then, one of the following cases holds for $(\M(r))$.
	\begin{enumerate}[(i)]
		\item Suppose that $\rho_1<\sqrt{\frac{\delta}{\gamma}}$. Then, $(\M(r))$ and $(\P(r))$ are unbounded problems with common optimal value $-\infty$.
		\item Suppose that $\rho_1=\sqrt{\frac{\delta}{\gamma}}$. Then, $(\M(r))$ and $(\P(r))$ have a common finite infimum that is equal to $-\frac{\b}{\g}$ but the infimum is not attained by a feasible solution in both problems.
		\item Suppose that $\rho_1>\sqrt{\frac{\delta}{\gamma}}$. Then, the unique optimal solution $(\sigma^\ast,\mu^\ast)$ of $(\M_\A)$ is also the unique optimal solution of $(\M(r))$ and the corresponding portfolio $w^\ast$ is the unique optimal solution of $(\P(r))$ if and only if $r\geq r^\ast$, where $r^\ast$ is defined by \eqref{rcond}. If $r<r^\ast$, then $(\hat{\sigma}(r),\hat{\mu}(r))$ is the unique optimal solution of $(\M(r))$ and
		\[
		\hat{w}(r)\coloneqq w(\hat{\mu})=\frac{1}{2\d}\sqb{\of{\frac{\d}{\g r+\b}-\g r-\b}C^{-1}m+\of{\a+\frac{\g r(\b r+\a)}{\g r+\b}}C^{-1}\1}
		\]
		is the unique optimal solution of $(\P(r))$.
	\end{enumerate}
\end{theorem}

\begin{proof}
	As in the proof of Theorem~\ref{thm1}, a point $(\sigma,\mu)\in\H\cap \L(r)$ satisfies \eqref{quadforint}, which reduces to a linear equation in $\sigma$ as we have $\rho_2=\sqrt{\frac{\delta}{\gamma}}$. Suppose that $r\neq-\frac{\b}{\g}$. Then, the unique solution of this equation is given by 
	\begin{equation}\label{hatsigma}
	\hat{\sigma}(r)=\frac{r+\frac{\b}{\g}}{2\rho_2}+\frac{\frac{1}{\g}}{2\rho_2\frac{\g}{\d}\of{r+\frac{\b}{\g}}}=\frac{(\g r+\b)^2+\d}{2\rho_2\g(r\g+\b)}=\frac{\g r^2+2\b r+\a}{2\rho_2(\g r+\b)}.
	\end{equation}
	The corresponding mean value of the point is given by
	\[
	\hat{\mu}(r)=\rho_2\hat{\sigma}-r=\frac{\g r^2+2\b r+\a-2(\g r+\b)r}{2(\g r+\b)}=\frac{\a-\g r^2}{2(\g r+\b)}.
	\]
	Hence, $\H\cap\L(r)=\cb{(\hat{\sigma}(r),\hat{\mu}(r))}$ as defined by \eqref{linsigma}. Moreover, from the third expression in \eqref{hatsigma}, it is clear that $\hat{\sigma}(r)<0$ if $r<-\frac{\b}{\g}$ and $\hat{\sigma}(r)>0$ if $r>-\frac{\b}{\g}$. If $r=-\frac{\b}{\g}$, then \eqref{quadforint} has no solution so that $\H\cap\L(r)\neq\emptyset$. It follows that $(\M(r))$ is feasible if and only if $r>-\frac{\b}{\g}$.
	
	For the rest of the proof, suppose that $r>-\frac{\b}{\g}$. Note that an analogue of the preparatory claim in the proof of Theorem~\ref{thm1} can be shown here with the same arguments: every $(\bar{\sigma},\bar{\mu})\in\H_+$ with $\bar{\sigma}\geq \hat{\sigma}(r)$ is also a feasible solution of $(\M(r))$. Similarly, (i) and (ii) here can be shown here by repeating the same arguments as in the proofs of (i) and (ii) of Theorem~\ref{thm1}. Hence, we consider only the case $\rho_1>\sqrt{\frac{\d}{\g}}$ here. As in the proof of (iii) of Theorem~\ref{thm3}, we note that the unique optimal solution $(\sigma^\ast,
	\mu^\ast)$ of $(\M_\A)$ is also the unique optimal solution of $(\M(r))$ if and only if $r\geq r^\ast$, where $r^\ast$ is defined by \eqref{rcond}. Since $\rho_1>\sqrt{\frac{\d}{\g}}=\rho_2$, we have $r^\ast>-\frac{\b}{\g}$. Suppose that $r<r^\ast$. In this case, $(\hat{\sigma}(r),\hat{\mu}(r))$ is the unique optimal solution of $(\M(r))$. This can be shown using similar arguments as in the proof of (ii)b. of Theorem~\ref{thm2}. To avoid repetitions, the details are omitted.
\end{proof}

\section{Conclusion}
In this paper, we look at the static portfolio optimization problem with two coherent risk measures. We consider the case where the asset returns take arbitrary joint distributions, and characterize the optimal portfolio through the associated dual problem. The dual problem is an infinite-dimensional convex optimization problem with a linear objective function and a linear constraint besides the convex constraints on some dual variables. We detect some special cases where the dual problem reduces to a linear programming problem. In the second part of the paper, under the restriction that asset returns are jointly Gaussian, we identify all parameter configurations under which an optimal portfolio exists and provide an explicit formula for it.

\section*{Acknowledgments}\label{Acknowledgements}

We are grateful to Oya Ekin Kara\c{s}an for her continuous support at the early stages of the project. In addition, we would like to thank two anonymous referees whose comments motivated us for much of the developments in Section~\ref{arbitrary}.

\bibliographystyle{named}

\end{document}